\documentclass[one column, draft,a4paper]{IEEEtran}

\usepackage{amsmath}
\usepackage{amsthm}
\usepackage{amsfonts}
\usepackage{amssymb}
\usepackage{amsbsy}
\usepackage{graphicx}
\usepackage{times}
\usepackage{default}
\usepackage{color}
\usepackage[noadjust]{cite}
\usepackage{enumerate}
\usepackage{tikz}
\usepackage{subcaption}
\usepackage[english]{babel}
\usepackage{lipsum}
\usepackage{soul}
\newcommand\blfootnote[1]{%
  \begingroup
  \renewcommand\thefootnote{}\footnote{#1}%
  \addtocounter{footnote}{-1}%
  \endgroup
}

\theoremstyle{definition}

\renewcommand{\tilde}{\widetilde}
\renewenvironment{pmatrix}{\lefto (\begin{matrix}}{\end{matrix}\right )}

\def\ba#1\ea{\begin{align*}#1\end{align*}}	
\def\ban#1\ean{\begin{align}#1\end{align}}	
\def\bac#1\eac{\vspace{\abovedisplayskip}{\par\centering$\begin{aligned}#1\end{aligned}$\par}\addvspace{\belowdisplayskip}}	

\newcommand{\lefto}{\mathopen{}\left}
\newcommand\mydots{\makebox[1em][c]{.\hfil.\hfil.}}

\newtheorem{theorem}{Theorem}

\newtheorem{lemma}{Lemma}
\newtheorem{definition}{Definition}
\newtheorem{proposition}{Proposition}

\newtheorem{remark}{Remark}

\IEEEoverridecommandlockouts 

\makeatletter

\newcommand{\vast}{\bBigg@{3}}
\newcommand{\Vast}{\bBigg@{4}}
\makeatother

\title{Canonical Conditions for \\ \vspace{-9pt} K/2 Degrees of Freedom }

\author{\IEEEauthorblockN{Recep G\"{u}l$^1$, David Stotz$^2$, Syed Ali Jafar$^3$, Helmut B\"{o}lcskei$^4$, and Shlomo Shamai (Shitz)$^5$} \newline
\IEEEauthorblockA{$^1$ Dept. IT \& EE, ETH Z\"urich, Z\"urich, Switzerland \\
$^2$ Kantonsschule Schaffhausen, Schaffhausen, Switzerland \\
$^3$ University of California at  Irvine, CA, USA \\
$^4$ Dept. IT \& EE and Dept. Math., ETH Z\"urich, Z\"urich, Switzerland \\
 $^5$   Dept. EE, Technion-Israel Institute of Technology, Haifa, Israel \\
}}

\begin{document}
\maketitle

\begin{abstract}
 We present a necessary and sufficient condition for $1/2$ degree of freedom for each user in constant $K$-user single-antenna interference channels. This condition applies to all channel topologies, i.e., to fully-connected channels as well as channels that have individual links absent, reflected by corresponding zeros in the channel matrix. Moreover, it captures the essence of interference alignment by virtue of being expressed in terms of a generic injectivity condition that guarantees separability of signal and interference. Finally, we provide codebook constructions achieving $1/2$ degree of freedom for each user for all channel matrices satisfying our condition.
\end{abstract}
\section{INTRODUCTION}
\blfootnote{The material in this paper was presented in part at the 2016 and the 2018 IEEE International Symposium on Information Theory (ISIT). \par The work of S. Shamai was supported by the
European Union's Horizon 2020 Research And Innovation Programme,
grant agreement no. 694630.} 

Cadambe and Jafar \cite{Cadambe08,Jafar11} proposed a signaling scheme---known as interference alignment---that exploits time-frequency selectivity to achieve $K/2$ degrees of freedom (DoF) in $K$-user single-antenna interference channels (ICs). In \cite{Etkin09} and \cite{Motahari14} it was shown that $K/2$ DoF can also be achieved in ICs with constant channel matrix, i.e, in the absence of channel selectivity. Gou et al. \cite{Gou11} furthermore demonstrated that even in the finite-state compound constant channel setting $K/2$ DoF can be achieved. Wu et al. \cite{Wu15} developed a general formula for the number of DoF in single-antenna ICs, extended to vector ICs in \cite{Stotz162}. This formula can, however, be difficult to evaluate as it is expressed in terms of R\'enyi information dimension \cite{Renyi59}. Building on the work by Wu et al. \cite{Wu15} and a breakthrough result in fractal geometry by Hochman \cite{Hochman14}, Stotz and B\"olcskei \cite{Stotz16} derived a DoF-formula for single-antenna ICs with constant channel matrix which is exclusively in terms of Shannon entropy. Based on this formula, the present paper establishes a necessary and sufficient condition for $1/2$ DoF for each user in constant $K$-user single-antenna ICs. This condition captures the essence of interference alignment by virtue of being expressed in terms of a generic injectivity condition that guarantees separability of signal and interference.  \par

\textit{Relation to prior work.} It was shown in \cite{Madsen05} that the number of DoF in a $K$-user fully-connected IC\footnote{Throughout the paper IC refers to a constant single-antenna IC.} is upper bounded by $K/2$. What is more, almost all IC matrices allow $K/2$ DoF, albeit an explicit characterization of this almost all set does not seem to be available \cite{Motahari14}. It is known \cite{Etkin09}, though, that i) $K/2$ DoF cannot be achieved if all elements of the IC matrix are rational numbers and  ii) $K/2$ DoF can be achieved if the diagonal\footnote{Throughout the paper, the term diagonal actually refers to the main diagonal.} elements of the IC matrix are irrational algebraic numbers and the off-diagonal elements are rational numbers. Further algebraic conditions on IC matrices to allow $K/2$ DoF were identified in \cite[Th.~7 and Th.~8]{Wu15}. An explicit and almost sure sufficient condition for $K/2$ DoF was reported in \cite{Stotz16}; the proof of this result rests on requiring linear independence---over the rational numbers---of monomials in the off-diagonal channel coefficients. We note, however, that the algebraic nature of these conditions renders them somewhat brittle. 
\par
\textit{Contributions.}
The necessary and sufficient condition for $1/2$ DoF for each user we obtain applies to all channel topologies, i.e., to fully-connected channels as well as channels that have individual links absent, reflected by corresponding zeros in the IC matrix. Moreover, we provide codebook constructions achieving full DoF in all ICs satisfying our condition.

\textit{Notation.} We use uppercase letters for random variables, and lowercase letters for deterministic quantities. Matrices are represented by boldface uppercase letters, and sets by calligraphic letters. For $x \in \mathbb{R}$, $\lfloor x \rfloor$ denotes the largest integer not exceeding $x$, and $|x|$ refers to the absolute value of $x$. $\mathbb{N}$, $\mathbb{Q}$, and $\mathbb{R}$ stand for the natural numbers including zero, the rational numbers, and the real numbers, respectively. Entropy is denoted by $H(\cdot)$, differential entropy by $h(\cdot)$, $d (\cdot)$ refers to Rényi information dimension \cite{Renyi59}, and $I(X;Y)$ is the mutual information between the random variables $X$ and $Y$. $\mathbb{E}$ designates the expectation operator.

\textit{Outline of the paper.} In Section \ref{secsystem}, we introduce the system model and we state definitions needed throughout the paper. Section \ref{secmain} presents the main result. In Section \ref{sec:ideas}, we develop the mathematical tools required in the proof of its sufficiency part, which, in turn, is established in Section \ref{secsuffproof}. In Section \ref{secbalancing}, we introduce the ``entropy balancing" idea underlying the proof---presented in Section \ref{proofthm1}---of the necessity part. In Section \ref{section::3UserNonFullyConnected}, we present a strengthened version of the necessity part of our main result for the $3$-user case. Section \ref{sec:applications} provides a corollary of our main statement, which, we feel, is of independent interest. The appendices collect various technical results.

\section{SYSTEM MODEL} \label{secsystem}
We consider a single-antenna $K$-user IC with $K \geqslant 3$, channel matrix $\mathbf{H}=(h_{ij})_{1\leqslant i,j\leqslant K}\in \mathbb{R}^{K\times K}$, and input-output relation 
\begin{equation*}
Y_{i}=\sqrt{\mathsf{snr}}\sum_{j=1}^{K}h_{ij}X_{j}+Z_{i},\,\, \, \, \, i=1,\!\mydots, K,
\end{equation*}
where $X_i \in \mathbb{R}$ is the input at the $i$-th transmitter, $Y_i \in \mathbb{R}$ is the output at the $i$-th receiver, and $Z_i \in \mathbb{R}$ is noise of absolutely continuous distribution satisfying $h(Z_i) > - \infty$ and $H(\lfloor Z_i \rfloor ) < \infty$. The $X_i$ are assumed to be statistically independent across transmitters and noise is taken to be i.i.d. across receivers and channel uses. Note that we exclude the case $K=2$, as here both users can achieve $1/2$ DoF simply through time sharing. We restrict ourselves to real-valued signals and channel matrices for simplicity of exposition.  An extension to the complex-valued case is possible based on a formula for the DoF of complex channel matrices provided in \cite{Stotz162}. The IC matrix $\mathbf H$ is assumed to be  known perfectly at all transmitters and receivers and we  take $h_{ii}\neq 0$, for $i=1,\mydots,K$, to avoid situations where direct links between  transmitter-receiver  pairs are absent. $\mathbf H$ is said to be fully connected if $h_{ij} \neq 0$, for all $i,j=1,\mydots, K$.  \par

We impose the average power constraint
 \begin{equation*} 
\mathbb{E} \! \left[X_i^{2} \right] \leqslant 1, \hspace{4mm} i=1, \dots, K, 
\end{equation*}
 and define the total number of DoF of the IC with channel matrix $\mathbf{H}$ as  
 \begin{equation} \label{definitionTotalDoF}
\mathsf{DoF}(\mathbf{H}):=\mathop{\lim\sup}_{\mathsf{snr}\rightarrow\infty}\frac{\overline{C}(\mathbf{H};\mathsf{snr})}{\frac{1}{2}\log \mathsf{snr}},
\end{equation}
where $\overline{C}(\mathbf{H};\mathsf{snr})$ stands for the sum-capacity of the IC. To define the DoF achieved by individual users, we start from the following multi-letter bound on the total number of DoF \cite[Eq. 137]{Wu15}. For all $\varepsilon >0$ and sufficiently large $m$, there exist $n \in \mathbb{N}$ and independent $n$-sequences of random variables $X_1^n, \mydots, X_K^n$ with $H(\lfloor X^n_i \rfloor) < \infty$, $i=1,\mydots, K$, such that
\begin{equation}  \label{rateDoFInequality}
\text{DoF}(\mathbf{H}) \leqslant  \frac{1}{mn}  \sum_{i=1}^{K}  \left\{ H \! \left( \left[ \sum_{j=1}^{K} h_{ij} X^n_j \right]_{\! m} \right)-H \!\left( \left[ \sum_{j \neq i}^{K} h_{ij} X^n_j \right]_{\! m} \right) \right\} +\varepsilon ,
\end{equation} 
where  $[x]_m= \frac{\lfloor 2^m x \rfloor}{2^m}$. We now define the number of DoF  of  user $i$ in $\mathbf{H}$ as
\begin{equation} \label{rateDoFDefinition}
\text{DoF}_i := \frac{1}{mn}  \left\{  H \! \left( \left[ \sum_{j=1}^{K} h_{ij} X^n_j \right]_{\! m} \right)-H \! \left( \left[ \sum_{j \neq i}^{K} h_{ij} X^n_j \right]_{\! m} \right)  \right\} ,
\end{equation} 
 for $i=1, \mydots, K$. Noting that for $\varepsilon, m,$ and $n$ in (\ref{rateDoFInequality}), the following holds \cite[Eq. 126]{Wu15}
 \begin{equation}  \label{rateDoFInequality2}
\text{DoF}(\mathbf{H})  \geqslant   \frac{1}{mn}  \sum_{i=1}^{K} \left\{  H \left( \left[ \sum_{j=1}^{K} h_{ij} X^n_j \right]_m \right)-H \left( \left[ \sum_{j \neq i}^{K} h_{ij} X^n_j \right]_m \right) \right\}  -\varepsilon,
\end{equation} 
combining  (\ref{rateDoFInequality}) and  (\ref{rateDoFInequality2}), and letting $\varepsilon \rightarrow 0$ yields
\begin{equation*} \text{DoF}(\mathbf{H}) = \sum_{i=1}^{K} \text{DoF}_i. \end{equation*}
We note that, thanks to Lemma \ref{lemma::1/2DoFMeansAtMostK/2Total} in Appendix \ref{subsection::ImplicationsOfRequiring1.2DoF}, if each user is to achieve at least $1/2$ DoF, i.e., $\text{DoF}_i \geqslant 1/2, i=1, \mydots, K$, then all input distributions yield the same $\text{DoF}_i, i=1, \mydots, K$.

  \par


\par 
 \section{MAIN RESULT} \label{secmain}
Before stating the main result, we need to introduce some formalisms. Denote the vector containing the off-diagonal elements of $\mathbf H$ by $\mathbf  {\widehat h}\in \mathbb R^{K(K-1)}$, and  let $f_1,f_2, ...$ be the monomials in the entries of $\mathbf {\widehat h}$ as follows:   $f_1,...,f_{\varphi(d)}$  are  the monomials of degree\footnote{ The ``degree'' of the monomial $x_1^{\alpha_1}x_2^{\alpha_{2}} \mydots \hspace{0.9mm} x_n^{\alpha_n}$ in the variables $x_1, x_2,  \mydots, x_n$ is given by $\alpha_1 + \alpha_2 +\mydots + \alpha_n$.} not larger than $d$, with $d \geqslant 0$, where  \cite[Ch.~9, Lem.~4]{Cox07} 
\begin{equation} \label{combinatorialPhi} {\varphi(d)}:=\binom{K(K-1)+d}{d}. \end{equation}  
\par 
 \begin{definition}(Scaling of IC matrices) \label{defScaling}
 We say that $\tilde{\mathbf{H}}$ is a scaled version of $\mathbf{H}$ \cite{Etkin09}, if $\tilde{\mathbf{H}}$ can be obtained by scaling (by nonzero real numbers) rows and columns of $\mathbf{H}$ according to $\tilde{\mathbf{H}}=\mathbf{D}\mathbf{H} \mathbf{D}'$, where $\mathbf{D}$ and $\mathbf{D}'$ are diagonal matrices with nonzero diagonal entries. 
 \end{definition} 
 
\begin{definition}(Channel topology) \label{def::topology}
Consider an IC with channel matrix $\mathbf{H}$. The topology of the IC is determined by the locations of the zeros, referred to as zero-set, in $\mathbf{H}$. Specifically, $h_{ij}=0$ reflects the absence of a link between transmitter $j$ and receiver $i$. 
\end{definition}

 
We proceed to stating our main result, a necessary and sufficient condition on the IC matrix $\mathbf{H}$ to allow for $1/2$ DoF for each user. In addition, for ICs satisfying this condition, we provide an explicit construction of codebooks achieving $1/2$ DoF for each user. Throughout the paper ``achieving (or allowing) $1/2$ DoF'' will frequently mean that at least $1/2$ DoF is achieved (or allowed for). For simplicity of exposition, we shall not distinguish between these cases and the case where exactly $1/2$ DoF is achieved (or allowed for). 
\begin{theorem}\label{thm1} 
Consider a $K$-user IC with channel matrix $\mathbf{H}$, and let $\varphi(d)$ be as in (\ref{combinatorialPhi}).
For $d \geqslant 0$ and $N \geqslant 1$, let 
\ba \mathcal W^{(\mathbf{H})}_{N,d} := \left \{\sum_{i=1}^{\varphi(d)} a_i f_i( \hat{\mathbf h}) \; : \; a_1,..., a_{\varphi (d)} \in \{0,...,N-1\} \right \} 
\ea
and set
\ban 	\mathcal W^{(\mathbf{H})}:= \bigcup_{d\geqslant 0}\bigcup_{N\geqslant 1} \mathcal W^{(\mathbf{H})}_{N,d}.	\label{wset} \ean
For each user to achieve  $1/2$ DoF, the following condition is sufficient for all and necessary for almost all IC matrices $\mathbf{H}$:  Either the IC matrix $\mathbf{H}$ itself or at least one scaled---in the sense of Definition \ref{defScaling}---version thereof satisfies the following condition: \newline \indent
For each $i=1, \! \mydots, K$, the map
\begin{equation*} \setlength{\abovedisplayskip}{8pt} \setlength{\belowdisplayskip}{7pt}
\qquad \qquad \qquad \qquad
\begin{matrix}
\mathcal{W^{(\mathbf{H})}}\times \mathcal{W^{(\mathbf{H})}}\rightarrow \mathcal{W^{(\mathbf{H})}}+h_{ii}\mathcal{W^{(\mathbf{H})}}\\[4pt]
(w_{1}, w_{2})\mapsto w_{1}+h_{ii}w_{2}
\end{matrix}\qquad\qquad \qquad 
( \ast)
\end{equation*} 
\indent is injective. \par
\textit{Proof}. See Section \ref{proofthm1}. 
\end{theorem} \par
\begin{remark} \label{remark1} 
As made explicit in the proof, the condition in Theorem \ref{thm1} is sufficient for ``all''---with respect to the channel coefficients\footnote{Throughout the paper, almost all sets are understood to be with respect to the channel coefficients.}---channel matrices $\mathbf{H}$ and necessary
for ``almost all" ICs. We hasten to add that we do not know of a way to test whether a given fully-connected $\mathbf{H}$ falls into the corresponding measure-zero set of exceptions. Finally, we note that
combinatorial and number-theoretic arguments can be used to establish necessity in the $3$-user case for ``all'' non-fully-connected ICs. This proof is very explicit and pedestrian and it seems unclear how the corresponding arguments could be extended to the $K$-user case.
\end{remark}
\begin{remark} \label{remarkDoFScaling}
We note that for all IC matrices $\mathbf{H}$, scaling according to Definition \ref{defScaling} does not change the total number of DoF \cite[Lemma 1]{Etkin09}. 
\end{remark}

\begin{remark} \label{remarkChannel}
Thanks to $\mathcal{W^{(\mathbf{H})}}$ being made up of integer linear combinations of monomials in the off-diagonal elements of $\mathbf{H}$, Condition ($ \ast$) is exclusively in terms of channel coefficients $h_{ij}$. 
\end{remark}

\begin{remark} \label{remarkEquivalentFormulation}
An equivalent formulation of Condition~($*$) that will turn out useful later is as follows: For each  $i=1,...,K$ and for all nonzero polynomials $P$, $Q$ in the variables $h_{ij}$, $i\neq  j$, with integer coefficients
 \begin{equation} \label{alternativeFormulationOfConditionStar}h_{ii}Q-P\neq 0. \end{equation}
To establish this equivalence, we first note that Condition~($*$) can alternatively be expressed as follows: For all $(w_1,w_2) \neq (\tilde{w}_1, \tilde{w}_2)$, 
\begin{equation} \label{rem3}
h_{ii } (w_2 - \tilde{w}_2) -(\tilde{w}_1-w_1) \neq 0.
\end{equation}
The argument is concluded by realizing that $w_2- \tilde{w}_2$ and $\tilde{w}_1-w_1$ are, by construction, polynomials in $h_{ij}, i \neq j$, with integer coefficients.  
\end{remark} \par

\section{Proof Idea and Auxiliary Results for Sufficiency}\label{sec:ideas}

We first describe the central ideas behind the proof of the sufficiency part of the statement in Theorem~\ref{thm1}. Recall that this part of the statement applies to ``all" IC matrices $\mathbf{H}$. The tenets of the proof are i) a condition guaranteeing separability of signal and interference and ii) an alignment cardinality constraint for the codebook. 
We begin  by restating a result from \cite{Stotz16}  needed in the proof.
\begin{proposition}[A simple variation of Prop.~1 in \cite{Stotz16}]\label{prop:apply}
Let $r \in (0,1)$  and let $\Psi_1,...,\Psi_K$ be independent discrete random variables. Consider the self-similar inputs $X_i= \sum_{k=0}^{\infty} r^k \Psi_{i,k}$, $i=1,...,K$, where $\{\Psi_{i,k} : k \geqslant 0 \}$ are independent copies of $\Psi_i$. Then, for all $\mathbf H$, 
\ban \text{DoF}_i \geqslant \min \lefto\{\frac{H\lefto (\sum_{j=1}^Kh_{ij}\Psi_{j}\right )}{\log (1/  r)}, 1\right  \} - \min \lefto \{\frac{H\lefto (\sum_{j\neq i}^K h_{ij}\Psi_{j}\right )}{\log (1/  r)}, 1\right \} \label{eq:apply}.  \ean

\end{proposition}
\begin{proof}
Follows directly from \cite[Eqs. 26-27]{Stotz16} and \cite[Eqs. 8-10]{Stotz16}.
\end{proof}

The strategy for proving the sufficiency part of Theorem~\ref{thm1} will be to employ Proposition~\ref{prop:apply} with the $\Psi_i$ i.i.d.\ uniform on $\mathcal W^{(\mathbf{H})}_{N^d,d}$ for some $d\geqslant 0$, $N> K-1$, and to show that the corresponding expression on the right-hand side (RHS) of \eqref{eq:apply} can be made to be arbitrarily close to $1/2$ for all $i=1,...,K$ concurrently. Henceforth, we shall drop the superscript in $\mathcal W^{(\mathbf{H})}_{N^d,d}$ whenever it is clear from the context. The first step in the proof establishes that $H(\sum_{j\neq i}^K h_{ij}\Psi_{j})$, i.e., the entropy of the interference part of the signal received at user $i$, can not become too large relative to the entropy of the signal component given by $H(\Psi_{i})$. This will be accomplished by noting that $\sum_{j\neq i}^K h_{ij}\Psi_{j}\in \mathcal W_{(K-1)N^d,d+1}\subseteq \mathcal W _{N^{d+1},d+1}$, where the inclusion is thanks to $N>K-1$, and then showing that the ratio $\frac{\log |\mathcal W_{N^{d+1}, d+1}|}{\log |\mathcal W_{N^{d}, d}|}$, where the denominator is equal to $H(\Psi_i)$ (owing to the assumption of $\Psi_i$ being uniformly distributed on $\mathcal W_{N^d,d}$), is close to $1$ for $d$ sufficiently large. This ``alignment cardinality" result is formalized as follows.  \par
\begin{lemma}\label{lem} 
For $N>1$, we have
\ban \liminf_{d\to \infty} \frac{\log |\mathcal W_{N^{d+1}, d+1}|}{\log |\mathcal W_{N^{d}, d}|}=1 .  \label{eq:liminf} \ean
\end{lemma}
\begin{proof}Since $\mathcal W_{N,d}\subseteq \mathcal W_{N',d'}$ for $N\leqslant N'$ and $d\leqslant d'$, we have $\log |\mathcal W_{N^{d+1}, d+1}| \geqslant \log |\mathcal W_{N^{d}, d}|$, which implies
\ba \liminf_{d\to \infty} \frac{\log |\mathcal W_{N^{d+1}, d+1}|}{\log |\mathcal W_{N^{d}, d}|}\geqslant 1 . \ea  
We establish \eqref{eq:liminf} by way of contradiction. To this end, assume that
\ba \liminf_{d\to \infty} \frac{\log |\mathcal W_{N^{d+1}, d+1}|}{\log |\mathcal W_{N^{d}, d}|}>1+\varepsilon , \ea
for some $\varepsilon>0$. Then, there exists a $d_0\geqslant 2$ such that for all $d\in \mathbb N$
\ban \frac{\log |\mathcal W_{N^{d_0+d+1}, d_0 +d +1}|}{\log |\mathcal W_{N^{d_0 +d}, d_0+d}|}>1+\varepsilon . \label{eq:repeated}\ean 
Repeated application of \eqref{eq:repeated} yields
\ba \log |\mathcal W_{N^{d_0}, d_0}| &< \frac{\log |\mathcal W_{N^{d_0+1}, d_0+1}|}{1+\varepsilon}\\ &< \ldots \\ &<  \frac{\log |\mathcal W_{N^{d_0+d}, d_0+d}|}{(1+\varepsilon)^d}. \ea 
Since $ |\mathcal W_{N^{d_0+d}, d_0+d}| \leqslant (N^{(d_0+d)} )^{\varphi(d_0+d) }$, we get 
\ban \log |\mathcal W_{N^{d_0}, d_0}| < \frac{\varphi(d_0+d)  (d_0+d) \log N}{(1+\varepsilon)^d} , 
\label{eq:toshow} \ean 
for all $d \in \mathbb{N}$. Next, note that owing to the assumption $N>1$, and thanks to $d_0 \geqslant 2$, it follows that $N^{d_0}-1>1$. Further, since $\{0,...,N^{d_0}-1\}\subseteq \mathcal W^{(\mathbf{H})}_{N^{d_0},d_0}$ for all $\mathbf{H}$, we have $|\mathcal W_{N^{d_0}, d_0}|>1$ and hence $ \log |\mathcal W_{N^{d_0}, d_0}| >0 $. The proof will be completed by establishing the contradiction $\log  |\mathcal W_{N^{d_0}, d_0}| =0$. This is accomplished by showing that  the RHS of \eqref{eq:toshow} tends to zero as $d\to \infty$.
To this end, we first note that 
\ban 	\varphi(d_0+d)&=\frac{(K(K-1) +d_0+d)!}{(K(K-1))! (d_0+d)!}\\ &=\frac{(K(K-1)+d_0+d) \cdot \ldots \cdot (d_0+d+1)}{(K(K-1))!}\\ &\leqslant \frac{(K(K-1)+d_0+d)^{K(K-1)}}{(K(K-1))!}	 \label{eq:last} \ean
and, since the largest power of $d$ occuring in the numerator of \eqref{eq:last} is $d^{K(K-1)}$, we find that
\ban \varphi(d_0+d) (d_0+d)\leqslant d^{K(K-1)+2} \label{eq:combine1} \ean  for sufficiently large $d$.
On the other hand, it follows from $e^x \geqslant x^{K(K-1)+3} /(K(K-1)+3)!$, for $x\geqslant 0$, that
\ban (1+\varepsilon)^d  = e^{d \ln (1+\varepsilon)} \geqslant \frac{(d\ln (1+\varepsilon))^{K(K-1)+3}}{(K(K-1)+3)!}  .\label{eq:combine2}\ean
Combining \eqref{eq:toshow}, \eqref{eq:combine1}, and \eqref{eq:combine2}, we finally obtain
\ba  0 & < \lim_{d\to\infty}	\frac{\varphi(d_0+d)  (d_0+d)}{(1+\varepsilon)^d}\\ &\leqslant \lim_{d\to\infty} \frac{d^{K(K-1)+2}(K(K-1)+3)! }{(d\ln (1+\varepsilon))^{K(K-1)+3}}\\ &=0,	\ea 
which completes the proof.
\end{proof}

\begin{remark}
Lemma~\ref{lem} above is inspired by \cite[Lem.~3]{RKJY14}.
\end{remark}

The second step in the proof will be concerned with the separability of signal and interference and uses the injectivity of the map in Condition~($*$) to establish that
\ban 
H \Bigg (h_{ii} \Psi_{i}+ \sum_{j\neq i}^K h_{ij}\Psi_{j}	\Bigg ) = H (h_{ii} \Psi_{i} ) + H \Bigg ( \sum_{j\neq i}^K h_{ij}\Psi_{j}	\Bigg ). \label{eq:sep}
\ean 
Upon resolving minor technicalities, Lemma \ref{lem} along with  (\ref{eq:sep}) will allow us to show that the RHS of \eqref{eq:apply} can be made to be arbitrarily close (from below) to $1/2$, for all  $i=1,...,K$ concurrently.

\section{PROOF OF SUFFICIENCY IN THEOREM \ref{thm1} } \label{secsuffproof}
Consider a channel matrix $\mathbf{H}$ satisfying Condition~($*$). As described in Section~\ref{sec:ideas}, the first step of the proof balances the entropies of the signal and interference components, $H (\Psi_{i})$ and $H \left ( \sum_{j\neq i}^K h_{ij}\Psi_{j}\right ) $, respectively. To this end, we choose $N>K-1$. As $K-1\geqslant 1$, we can apply Lemma~\ref{lem} to find a subsequence $\{\mathcal W_{N^{d_n},d_n} \}_{n\geqslant 0}$ of $\{\mathcal W_{N^d,d} \}_{d\geqslant 0}$ such that 
\ban \lim_{n\to \infty} \frac{\log |\mathcal W_{N^{d_n+1}, d_n+1}|}{\log |\mathcal W_{N^{d_n}, d_n}|}=1 . \label{eq:limit} \ean 
Next, consider the set of discrete random variables $\Psi_1^{(n)},...,\Psi_K^{(n)}$ distributed  i.i.d.\ uniformly on $\mathcal W_{N^{d_n},d_n}$ and construct the corresponding self-similar transmit signals
\begin{align*}
X_i = \sum_{k=0}^{\infty} r^k \Psi_{i,k}^{(n)},
\end{align*}
where the $\Psi_{i,k}^{(n)}$ are independent copies of $\Psi_i^{(n)}$, and $r \in (0,1)$. Using these codebooks, we now apply Proposition~\ref{prop:apply} with $r=|\mathcal W_{N^{d_n},d_n}|^{-2}$ to get that the $i$-th user, $i=1, \mydots, K$,  achieves 
\ban \min \lefto \{\frac{H\lefto (\sum_{j=1}^Kh_{ij}\Psi_{j}^{(n)}\right )}{2 \log |\mathcal W_{N^{d_n},d_n}|}, 1\right \} - \min \lefto \{\frac{H\lefto (\sum_{j\neq i}^K h_{ij}\Psi^{(n)}_{j}\right )}{2 \log |\mathcal W_{N^{d_n},d_n}|}, 1\right \}   \ean 
DoF, 
for $n\in \mathbb N$. 
Note that $\sum_{j\neq i}^K h_{ij}\Psi^{(n)}_{j}\in \mathcal W_{(K-1)N^{d_n}, d_n+1} \subseteq \mathcal W_{N^{d_n+1}, d_n+1}$ thanks to $N> K-1$. It follows from the cardinality bound for entropy that 
\ban \frac{H\lefto (\sum_{j\neq i}^K h_{ij}\Psi^{(n)}_{j}\right )}{2 \log |\mathcal W_{N^{d_n},d_n}|}\leqslant  \frac{\log |\mathcal W_{N^{d_n+1}, d_n+1}|}{2\log |\mathcal W_{N^{d_n}, d_n}|}\xrightarrow{n\to \infty} \frac{1}{2}, \label{conv}\ean 
where we used \eqref{eq:limit}.

The second step of the proof concerned with separability of signal and interference starts by establishing that
 \ban H \Bigg (h_{ii} \Psi^{(n)}_{i}+ \sum_{j\neq i}h_{ij}\Psi^{(n)}_{j}	\Bigg ) = H \Bigg (h_{ii} \Psi^{(n)}_{i}, \sum_{j\neq i}h_{ij}\Psi^{(n)}_{j}	\Bigg ) ,\label{eq:weshow}\ean
 for $i=1,...,K$. To this end, we apply the chain rule to find
\ban &H \Bigg (h_{ii} \Psi^{(n)}_{i}, \sum_{j\neq i}h_{ij}\Psi^{(n)}_{j}	\Bigg )  \label{eq:combined1}  \\&= H \Bigg (h_{ii} \Psi^{(n)}_{i}, \sum_{j\neq i}h_{ij}\Psi^{(n)}_{j} , 	h_{ii} \Psi^{(n)}_{i}+ \sum_{j\neq i}h_{ij}\Psi^{(n)}_{j}	\Bigg )\\ &=  H \Bigg (h_{ii} \Psi^{(n)}_{i}+ \sum_{j\neq i}h_{ij}\Psi^{(n)}_{j}	\Bigg ) \nonumber \\& \phantom{-} + H \Bigg (h_{ii} \Psi^{(n)}_{i}, \sum_{j\neq i}h_{ij}\Psi^{(n)}_{j}	 \; \!  \Bigg\vert   \; \!  h_{ii} \Psi^{(n)}_{i}+ \sum_{j\neq i}h_{ij}\Psi^{(n)}_{j}	\Bigg ). \label{eq:combined2} \ean
Next, we note that  injectivity of the map in Condition~($*$) implies \ban H \Bigg (h_{ii} \Psi^{(n)}_{i}, \sum_{j\neq i}h_{ij}\Psi^{(n)}_{j}	 \; \!  \Bigg\vert   \; \!  h_{ii} \Psi^{(n)}_{i}+ \sum_{j\neq i}h_{ij}\Psi^{(n)}_{j}	\Bigg )	=0 \label{eq:amounts}, \ean
which, when combined with \eqref{eq:combined1}--\eqref{eq:combined2}, yields \eqref{eq:weshow}.
From \eqref{eq:weshow} and the independence of the $\Psi_i^{(n)}$ across users $i =1, \mydots, K$, it  now follows that 
\ban \frac{H\lefto (\sum_{j=1}^Kh_{ij}\Psi^{(n)}_{j}\right ) - H\lefto (\sum_{j\neq i}^K h_{ij}\Psi^{(n)}_{j}\right ) }{2 \log |\mathcal W_{N^{d_n},d_n}|} &= \frac{H(\Psi^{(n)}_i)}{2 \log |\mathcal W_{N^{d_n},d_n}|}\label{eq:1/2a} \\ &= \frac{1}{2} , \label{eq:1/2b} \ean
where we used that $\Psi_i ^{(n)}$ is uniformly distributed on $ \mathcal W_{N^{d_n},d_n}$, for all $i=1,...,K$.
This allows us to conclude that, for all $n\in \mathbb N$, we have
\ban \nonumber  &\min \lefto \{\frac{H\lefto (\sum_{j=1}^Kh_{ij}\Psi^{(n)}_{j}\right )}{2 \log |\mathcal W_{N^{d_n},d_n}|}, 1\right \}  - \min \lefto \{\frac{H\lefto (\sum_{j\neq i}^K h_{ij}\Psi^{(n)}_{j}\right )}{2 \log |\mathcal W_{N^{d_n},d_n}|}, 1\right \} \\ &\geqslant 1-  \frac{\log |\mathcal W_{N^{d_n+1}, d_n+1}|}{2\log |\mathcal W_{N^{d_n}, d_n}|} ,  \label{eq:summing}  \ean 
 as either the first minimum on the left-hand side (LHS) of \eqref{eq:summing} coincides with the non-trivial term in which case by
\eqref{eq:1/2a} and \eqref{eq:1/2b} the second minimum also coincides with the non-trivial term, and therefore the LHS
of \eqref{eq:summing} equals $1/2\geqslant 1-  \frac{\log |\mathcal W_{N^{d_n+1}, d_n+1}|}{2\log |\mathcal W_{N^{d_n}, d_n}|}$, for all $n \in \mathbb{N}$, thanks to (\ref{conv}), or the first minimum coincides with $1$ in which case we upper-bound the second minimum according to $\frac{H\lefto (\sum_{j\neq i}^K h_{ij}\Psi^{(n)}_{j}\right )}{2 \log |\mathcal W_{N^{d_n},d_n}|}\leqslant  \frac{\log |\mathcal W_{N^{d_n+1}, d_n+1}|}{2\log |\mathcal W_{N^{d_n}, d_n}|} $ again using \eqref{conv}.
The proof is completed by  noting that the RHS  of \eqref{eq:summing} approaches $1/2$ (from below) as $n\to \infty$. We hasten to add that no restrictions had to be imposed on $\mathbf{H}$, so sufficiency in Theorem \ref{thm1} applies to ``all" channel matrices.

\par 
We have established that if $\mathbf{H}$ itself satisfies Condition~($\ast$), then each user achieves $1/2$ DoF. It remains to show that if at least one scaled version of $\mathbf{H}$ satisfies Condition~($\ast$), while $\mathbf{H}$ itself may or may not satisfy Condition~($\ast$), then each user achieves $1/2$ DoF in $\mathbf{H}$. To this end, let $\mathbf{H}$ be such that its scaled version $\mathbf{\tilde{H}}$ satisfies Condition~($\ast$). Further let, for some $\varepsilon \in (0,1/2)$ and $r \in (0,1)$, $\tilde{\Psi}_1, \mydots, \tilde{\Psi}_K$ be the random variables corresponding to the inputs $X_1, \mydots, X_K$ for $\mathbf{\tilde{H}}$ according to Proposition~\ref{prop:apply}. Then, by Proposition~\ref{prop:apply} and what was established above for $\mathbf{H}$, specifically (\ref{eq:summing}), user $i$ in $\mathbf{\tilde{H}}$ achieves 
\begin{equation} \label{etwasImMeer}
  \min \lefto\{\frac{H\lefto (\sum_{j=1}^K \tilde{h}_{ij}\tilde{\Psi}_{j}\right )}{\log (1/  r)}, 1\right  \} - \min \lefto \{\frac{H\lefto (\sum_{j\neq i}^K \tilde{h}_{ij} \tilde{\Psi}_{j}\right )}{\log (1/  r)}, 1\right \} \geqslant 1/2,
\end{equation} 
DoF, for $i=1, \mydots, K$. Next, note that $\mathbf{H}$ can be obtained as a scaled version of $\mathbf{\tilde{H}}$ according to
\begin{equation*}{\mathbf{H}}=\begin{pmatrix} p_1c_1 \tilde{h}_{11}&p_1c_2  \tilde{h}_{12}& \mydots & p_1c_K  \tilde{h}_{1K}\\ p_2c_1  \tilde{h}_{21}&p_2c_2  \tilde{h}_{22}& \mydots & p_2 c_K  \tilde{h}_{2K}\\ \mydots & \vdots & \vdots & \vdots \\ p_K c_1  \tilde{h}_{K1}&p_K c_2  \tilde{h}_{K2}& \mydots & p_K c_K  \tilde{h}_{KK}\end{pmatrix}\hspace{-2.6pt}, \end{equation*}
where $p_i, c_j$, $i, j = 1, \mydots, K$ are nonzero scalars. 
Applying Proposition \ref{prop:apply} with $\Psi_i= \tilde{\Psi}_i/c_i$, for $i=1, \mydots, K$, we conclude that for user $i$ in $\mathbf{H}$,
\begin{align}
\text{DoF}_i \geqslant & \min \lefto\{\frac{H\lefto (\sum_{j=1}^K p_i \tilde{h}_{ij}\tilde{\Psi}_{j}\right )}{\log (1/  r)}, 1\right  \} - \min \lefto \{\frac{H\lefto (\sum_{j\neq i}^K p_i  \tilde{h}_{ij}\tilde{\Psi}_{j}\right )}{\log (1/  r)}, 1\right \}  \\   \label{eqPreserv} =&
\min \lefto\{\frac{H\lefto (\sum_{j=1}^K  \tilde{h}_{ij}\tilde{\Psi}_{j}\right )}{\log (1/  r)}, 1\right  \} - \min \lefto \{\frac{H\lefto (\sum_{j\neq i}^K  \tilde{h}_{ij}\tilde{\Psi}_{j}\right )}{\log (1/  r)}, 1\right \} \geqslant \frac{1}{2}.
\end{align} 
We note that (\ref{eqPreserv}) holds as scaling a random variable by a constant does not change its entropy.

\section{ENTROPY BALANCING} \label{secbalancing}
We next establish ``balancing" results on the entropies of signal and interference contributions that will turn out instrumental in the proof of the necessary part of Theorem \ref{thm1}. To this end, we need the following preparatory result, which is a simple modification of \cite[Th. 3]{Stotz16}. 
\begin{proposition} \label{prop}
Consider the set $\mathcal{S}$ of all IC matrices $\mathbf{H}$ for which the following holds: For every $\varepsilon \in (0, 1/2)$, there exist independent discrete random variables $V_1, \! \mydots, V_K$ of finite entropy such that
\begin{align} \label{prop1}  \text{DoF}_i -\varepsilon  \leqslant  \frac { \Big [{ H \Big ({\sum _{j=1}^{K} h_{ij} V_{j} }\Big )- H \Big ({ \sum _{j\neq i}^{K} h_{ij} V_{j} }\Big )}\Big ]}{ \max _{i=1, \mydots ,K} H   \Big ({\sum _{j=1}^{K} h_{ij} V_{j} }\Big )} \end{align}
for $i=1, \! \mydots, K$. $\mathcal{S}$ is an almost all subset of $\mathbb{R}^{K \times K}$. 
\end{proposition} \par
\begin{proof}
See Appendix \ref{proofOfProp}.
\end{proof}
We are now ready to state our entropy balancing result. 
\begin{lemma} \label{lem1}
Let $\mathbf{H}$ be a fully-connected IC matrix contained in the a.a. set\footnote{Note that almost all fully-connected IC matrices are contained in $\mathcal{S}$ as fully-connected ICs constitute an a.a. set in $\mathbb{R}^{K \times K}$, and the intersection of two a.a. sets in $\mathbb{R}^{K \times K}$ is again an a.a. set in $\mathbb{R}^{K \times K}$.} $\mathcal{S}$  in Proposition \ref{prop} for DoF$_i=1/2$,  $i=1, \mydots, K$. For $\varepsilon \in (0,1/2)$, denote the corresponding discrete random variables satisfying (\ref{prop1})  by $V_1, \! \mydots,V_K$.  Then, we have
\begin{flalign}  \label{8}
& \frac{H \Big ({\sum _{j \neq i}^{K} h_{ij} V_{j} }\Big )}{  H  (V_i)} =1 + \mathcal{O}(\varepsilon) 
 \text{,  \space for{ } } i=1, \! \mydots, K,  \\[3.5mm] 
\label{9}
 & \frac{H   (V_i  )}{  H  (V_j)} =1 + \mathcal{O}(\varepsilon) \text{, \space for   }  i, j \in \{1,\! \mydots,K \}, i \neq j,\\[3.5mm]
 \label{22}
 & \frac{H \Big ({\sum _{j=1}^{K} h_{ij} V_{j} }\Big )}{H(V_i)}=2 + \mathcal{O}(\varepsilon)  \text{, \space for } i=1, \! \mydots, K. 
\end{flalign} 
\end{lemma} 
\begin{proof} 
Starting from (\ref{prop1}) with DoF$_i=1/2$ and rearranging terms, we get \vspace{2pt}
\begin{equation}  \label{12}
2 H \Bigg({\sum _{j \neq i}^{K} h_{ij} V_{j} }\Bigg) \leqslant (1+2 \varepsilon)  H \Bigg ({\sum _{j=1}^{K} h_{ij} V_{j} }\Bigg ),
\end{equation}
for $i=1, \mydots, K$. 
Invoking the following inequality valid for independent discrete random variables $X$, $Y$ \cite[Ex. 2.14]{Cover06}
\begin{equation} \label{wki} H  (X+Y) \leqslant H  (X ) + H  ( Y) \end{equation}
 on the RHS of (\ref{12}) yields
\begin{equation}  \label{13}
   (1-2\varepsilon) { H   \Bigg ({\sum _{j \neq i}^{K} h_{ij} V_{j} }\Bigg )} \leqslant (1+2 \varepsilon)  H    ( V_i ),
\end{equation}
 for $i=1, \mydots, K$. 
Next, we show that 
\begin{equation}  \label{14}
H   \Bigg ({\sum _{j \neq i}^{K} h_{ij} V_{j} }\Bigg) \geqslant \frac{(1-2\varepsilon)}{(1+2\varepsilon)}  H  (V_i),
\end{equation}  
for  $i=1, \mydots, K$. 
To this end, w.l.o.g., we assume that 
\begin{equation} \label{eqAssumeInequality}
H(V_1) \geqslant H(V_2) \geqslant \mydots \geqslant H(V_K).
\end{equation}
 Applying \cite[Ex. 2.14]{Cover06}
\begin{equation} \label{wki2}
H(\alpha X+\beta Y) \geqslant \max{\{H(X),H(Y)}\}
\end{equation} 
 for independent discrete random variables $X$ and $Y$, and arbitrary $\alpha$, $\beta$  $\in \mathbb{R} \setminus \{0\}$, with $X= \sum _{\substack{j \neq 1,i }}^{K} h_{ij} V_{j}$, $Y=V_1$, $\alpha=1$, $\beta= h_{i1}$, for $i=2, \! \mydots, K$, we obtain
\begin{equation}  \label{15}
(1-2\varepsilon) H(V_i) \leqslant (1-2\varepsilon) H(V_1) \leqslant (1+2\varepsilon)  { H   \Bigg ({\sum _{j \neq i}^{K} h_{ij} V_{j} }\Bigg )},
\end{equation}  
where the first inequality follows from (\ref{eqAssumeInequality}). 
This establishes (\ref{14}) for $i=2, \! \mydots, K$. The statement for the case $i=1$ is obtained as follows. First, note that
 \begin{equation} \label{16}
(1-2\varepsilon)  H(V_1) \leqslant  (1-2\varepsilon)  { H   \Bigg ({\sum _{j \neq i}^{K} h_{ij} V_{j} }\Bigg )} \leqslant (1+2 \varepsilon)  H   ( V_i  )\end{equation}  
 for all $i \neq 1$, where the first inequality is again by application of (\ref{wki2}), and the second is by (\ref{13}). Next, by application of (\ref{wki2}), we get 
\begin{equation} \label{eqForIEquals1}
(1+2 \varepsilon) H    ( V_i ) \leqslant (1+2\varepsilon) { H  \Bigg ({\sum _{j \neq 1}^{K} h_{1j} V_{j} }\Bigg )},
\end{equation} 
 for all $i \neq 1$.
Inserting (\ref{eqForIEquals1}) into (\ref{16}) yields 
\begin{equation} \label{eqForIEquals2}
(1-2 \varepsilon) H    ( V_1 ) \leqslant (1+2\varepsilon) { H  \Bigg ({\sum _{j \neq 1}^{K} h_{1j} V_{j} }\Bigg )},
\end{equation} 
which establishes (\ref{14}) for $i=1$. We can now combine (\ref{13}) and (\ref{14}) to get
\begin{equation} \label{17}
 \frac{1-2\varepsilon}{1+2\varepsilon} \leqslant \frac{H   \Big ({\sum _{j \neq i}^{K} h_{ij} V_{j} }\Big )}{H(V_i)} \leqslant \frac{1+2\varepsilon}{1-2\varepsilon},
\end{equation}
for all $i$, which establishes (\ref{8}). \par
To prove (\ref{9}), we again assume, w.l.o.g., that $H(V_1) \geqslant \mydots \geqslant H(V_K)$, and simply note that thanks to (\ref{16}) 
\begin{equation}  \label{18}
\frac{1-2\varepsilon}{1 +2 \varepsilon} \leqslant   \frac{H(V_K)}{H(V_1)} \leqslant \frac{H(V_i)}{H(V_j)} \leqslant \frac{H(V_1)}{H(V_K)} \leqslant \frac{1+2 \varepsilon}{1- 2 \varepsilon},
\end{equation}
for $i,j=1, \! \mydots, K$.  \par 
Finally, to establish (\ref{22}), we start by noting that
\begin{equation} \label{23}
\frac{H  \Big ({\sum _{j=1}^{K} h_{ij} V_{j} }\Big )}{H \mathopen {} ( V_{i} )}  \geqslant  \Big(\frac{1-2\varepsilon}{1+2\varepsilon} \Big)  \frac{H   \Big ({\sum _{j=1}^{K} h_{ij} V_{j} }\Big )}{H    \Big ({\sum _{j \neq i}^{K} h_{ij} V_{j} }\Big )}  \\  \geqslant  \frac{2 (1-2\varepsilon)}{(1+2\varepsilon)^2},
\end{equation}
owing to (\ref{17}) and (\ref{12}).
Using (\ref{wki}) and (\ref{17}), we get
\begin{equation} \label{24}  H   \Bigg ({\sum _{j=1}^{K} h_{ij} V_{j} }\Bigg ) \leqslant \Big( 1+ \frac{1+2\varepsilon}{1-2\varepsilon} \Big) H   ({ V_{i} }).
\end{equation}
Combining (\ref{24}) with (\ref{23}) then establishes (\ref{22}).
\end{proof} \par \vspace{5pt}

\section{PROOF OF NECESSITY IN THEOREM \ref{thm1}} \label{proofthm1}
We assume that the IC has at least three users. (Recall that we excluded the 2-user case, as here each user can achieve exactly $1/2$ DoF by time-sharing regardless of the underlying  $\mathbf{H}$-matrix). We first prove Theorem \ref{thm1} for fully-connected ICs. The proof is effected by contradiction.
Towards this contradiction, we assume that the fully-connected $\mathbf{H}$  is in the almost all set ${\cal S}$ in Proposition \ref{prop} corresponding to DoF$_i = 1/2$, $i = 1,2, \mydots, K$, while at the same time Condition~($ \ast$) is violated for $\mathbf{H}$ and all scaled versions thereof. In particular, Condition~($ \ast$) must also be violated for 
\begin{equation} \label{eq::scaledMatrixHTilde}{\tilde{\mathbf{H}}}=\begin{pmatrix}  h_{11}& h_{12}&  \mydots &h_{1(K-1)}&  1\\ 1&  h_{22}& \mydots & h_{2 (K-1)} &   1 \\ h_{31}&  h_{32}& \mydots &h_{3(K-1)}& 1 \\ \vdots & \vdots & \vdots & \vdots & \vdots \\ h_{(K-1)1}&h_{(K-1)2}& \mydots & h_{(K-1)(K-1)} &1 \\ 1&1& \mydots & 1 & h_{KK} 
\end{pmatrix}\hspace{-2.6pt}, \end{equation}
which can be obtained from $\mathbf{H}$ by scaling according to Definition \ref{thm1} as follows. Denoting the entries of $\mathbf{H}$ as $h'_{ij}$, $i,j=1, \mydots, K$, multiply rows $i \neq 2, i=1, \mydots, K-1$, in $\mathbf{H}$ by $\frac{h'_{2K}}{h'_{21}{h'_{iK}}}$, row $2$ by $\frac{1}{h'_{21}}$, and row $K$ by $\frac{1}{h'_{K1}}$. Then, multiply columns $j=1,\mydots, K-1$ by $\frac{h'_{K1}}{h'_{Kj}}$ and column $K$ by $\frac{h'_{21}}{h'_{2K}}$. 
The reduction to the specific $\mathbf{\tilde{H}}$ in (\ref{eq::scaledMatrixHTilde}) is made for simplicity of exposition. Thanks to Lemma \ref{lemmaAlmostAllSet}, $\mathbf{\tilde{H}}$ is also in the almost all set $\mathcal{S}$ in Proposition \ref{prop} corresponding to DoF$_i = 1/2$, $i = 1,2, \mydots, K$. As by assumption Condition~($\ast$) is violated for $\mathbf{H}$ and all scaled versions thereof, there must be a user $i$ such that, thanks to Remark \ref{remarkEquivalentFormulation}, there exist polynomials $P, Q \in \mathcal{W}^{(\mathbf{\tilde{H}})}$ so that 
\begin{equation} \label{eq::h_iiPQ}
h_{ii}= \frac{P}{Q},
\end{equation}
where $h_{ii}$ is the $i$-th diagonal entry of $\mathbf{\tilde{H}}$.  
As $\mathbf{\tilde{H}}$ is fully connected and contained in the set ${\cal S}$ in Proposition \ref{prop} for DoF$_i=1/2$, $i=1,2, \mydots, K$, it follows from (\ref{22}) that
\begin{equation} \label{eq::Ratio2BetweenSignalNoise}
\frac{H \left(h_{ii} V_i + \displaystyle \sum_{j \neq i} h_{ij} V_j \right)}{H(V_i)}=2 +\mathcal{O}(\varepsilon),
\end{equation}
where $V_1, V_2, \mydots, V_K$ are independent random variables satisfying
(\ref{prop1}) for DoF$_i=1/2$, $i=1, 2, \mydots, K$. We shall show that this leads to a contradiction by proving that (\ref{prop1}) with DoF$_i=1/2$, for $i=1, 2, \mydots, K$, and $V_1, V_2, \mydots, V_K$ satisfying (\ref{eq::Ratio2BetweenSignalNoise}), implies 
\begin{equation} \label{eq::InductionArgumentToProve}
\frac{H \left(h_{ii} V_i + \displaystyle \sum_{j \neq i} h_{ij} V_j \right)}{H(V_i)}=1 +\mathcal{O}(\varepsilon),
\end{equation}
for the index $i$ that satisfies (\ref{eq::h_iiPQ}).

The contradiction will be established through two nested inductive arguments. That is to say, the base case for the main induction argument 
over the maximum number of terms in $P$ and $Q$ in (\ref{eq::h_iiPQ})
will be effected by another induction, namely over the maximum degree of the polynomials, actually monomials in the base case, $P$ and $Q$.

\textit{Base case over the number of terms in $P$ and $Q$}. Let $p$ be the maximum of the number of terms in $P$ and $Q$. We start with the base case $p=1$, i.e.,
both $P$ and $Q$ are monomials. Then, we can express (\ref{eq::h_iiPQ}) as follows
\begin{equation} \label{eq::h_iiMonomialsPQ}
h_{ii}= \frac{a}{b} \frac{\displaystyle \prod_{j,k, j \neq k}^{K(K-1)}h_{jk}^{\alpha_{jk}}}{\displaystyle \prod_{j,k, j \neq k}^{K(K-1)}h_{jk}^{\beta_{jk}}},
\end{equation}
where $a, b \in \mathbb{Z}$, $\alpha_{jk}, \beta_{jk} \in \mathbb{N}$, and $h_{jk}$ are off-diagonal elements of $\mathbf{\tilde{H}}$ with $j,k \in \{1, \mydots, K\}, j \neq k$.

We establish (\ref{eq::InductionArgumentToProve}) by induction over $d= \max \{ \sum_{j,k, j \neq k} \alpha_{jk},\sum_{j,k, j \neq k} \beta_{jk}\}$ with the base case $d=1$. First, the case $d=0$ is dealt with separately by showing that
\begin{equation} \label{eq::Degree0CaseSignalNoise1}
    \frac{H   \left( a V_i + b \sum_{j \neq i} h_{ij} V_j \right)}{H(V_i)}=1+\mathcal{O}(\varepsilon). 
\end{equation}
From (\ref{8}) with $i=K$ and using the specific form of $\mathbf{\tilde{H}}$ in (\ref{eq::scaledMatrixHTilde}), we get 
\begin{equation} \label{eq::Degree0CaseSumOfVi1LineK}
 \frac{H   \left( \sum_{j=1}^{K-1} V_j \right)}{H(V_K)}=1+\mathcal{O}(\varepsilon).
\end{equation}
Again by (\ref{8}), but now with $i=2$, we have
\begin{equation} \label{eq::Degree0CaseSumOfVi1Line2}
 \frac{H   \left( V_1+ \sum_{j=3}^{K} h_{2j}V_j \right)}{H(V_2)}=1+\mathcal{O}(\varepsilon),
\end{equation}
where we used that $h_{21}=1$. Next, note that 
\begin{equation} \label{eq::Degree0CaseSumOfVi1Line2Preparation}
 1+\mathcal{O}(\varepsilon)= \frac{H   \left( V_1+ \sum_{j=3}^{K} h_{2j}V_j \right)}{H(V_2)} \geqslant \frac{H   \left( V_1+ V_K \right)}{H(V_2)} \geqslant \frac{H   \left( V_1 \right)}{H(V_2)}=1+\mathcal{O}(\varepsilon),
\end{equation}
where both inequalities follow from (\ref{wki2}), and the last equation is by (\ref{9}). This yields 
\begin{equation} \label{eq::Degree0CaseSumOfVi1Line2Cleaner}
 \frac{H   \left( V_1 +V_K \right)}{H(V_2)}= 1+\mathcal{O}(\varepsilon).
\end{equation}
We can now replace the denominators of (\ref{eq::Degree0CaseSumOfVi1LineK}) and (\ref{eq::Degree0CaseSumOfVi1Line2Cleaner}) by $H(V_1)$, by applying (\ref{9}) with $i=1$ and $j=K$, and $i=1$ and $j=2$, respectively. Then, Lemma \ref{lemtao} with $X=V_1$, $Y_1=\sum_{j \neq 1}^{K-1} V_j$, and $Y_2=V_K$ yields
\begin{equation} \label{eq::Degree0CaseSumOfVi1Conclusion}
\frac{ H \left( \sum_{j=1}^K V_j  \right)}{H(V_1)}= 1+\mathcal{O}(\varepsilon).
\end{equation}
Next, thanks to (\ref{eq::Degree0CaseSumOfVi1Conclusion}), (\ref{9}), and (\ref{wki2}), we have 
\begin{equation} \label{eq::Degree0CaseSumOfVi1Helper}
 1+\mathcal{O}(\varepsilon)= \frac{H \left( \sum_{j= 1}^K V_j  \right)}{H(V_i)}  \geqslant \frac{H \left( \sum_{j \neq i} V_j  \right)}{H(V_i)} \geqslant  1 + \mathcal{O}(\varepsilon),
\end{equation}
for all $i=1, \mydots, K$.
Applying \cite[Th. 14]{Wu15} (see Appendix \ref{appEntropyDiff}) with $p=a, q=b, X=V_i$, and $Y=\sum_{j \neq i} V_j$, and dividing the result thereof by $H(V_i)$ yields
 \begin{align}\label{eq::Degree0CasePunchLineInequality} 
\frac{H \left(aV_i+b\sum_{j \neq i} V_j \right)}{H(V_i)}  -\frac{ H \left(\sum_{j=1}^K V_j \right)}{H(V_i)} \leqslant   \tau_{a,b} \left (\frac{2 H \left(\sum_{j=1}^K V_j \right)-H(V_i)-H \left(\sum_{j \neq i} V_j \right)}{H(V_i)} \right ),
\end{align}
where $\tau_{a,b}= 7 \lfloor \log |a| \rfloor + 7 \lfloor \log |b| \rfloor +2$. Thanks to (\ref{eq::Degree0CaseSumOfVi1Helper}), the RHS of (\ref{eq::Degree0CasePunchLineInequality}) equals $\mathcal{O}(\varepsilon)$. We therefore have 
\begin{equation} 
1 \leqslant \frac{H \left(aV_i+b\sum_{j \neq i} V_j \right)}{H(V_i)} \leqslant 1+ \mathcal{O}(\varepsilon),
\end{equation}
where the first inequality follows from (\ref{wki2}). In summary, we get
\begin{equation} \label{eq::Degree0CasePunchLineFinalForK}
\frac{H \left(aV_i+b\sum_{j \neq i} V_j \right)}{H(V_i)} = 1+ \mathcal{O}(\varepsilon),
\end{equation}
which, upon noting that $h_{K1}=h_{K2}= \mydots = h_{K(K-1)}=1$, establishes (\ref{eq::Degree0CaseSignalNoise1}) for $i=K$. For $i \neq K$,  we apply (\ref{8}) to obtain
\begin{equation} \label{eq::Degree0CaseIfNotUserK}
\frac{H \left(b V_K + b \sum_{j \neq i}^{K-1} h_{ij}V_j \right)}{H(V_i)} = 1+ \mathcal{O}(\varepsilon).
\end{equation}
Furthermore, thanks to (\ref{eq::Degree0CasePunchLineFinalForK}) and $(\ref{wki2})$, we have
\begin{equation}  \label{eq::Degree0CaseIfNotUserKNo2}
\frac{H \left(aV_i+bV_K \right)}{H(V_i)} = 1+ \mathcal{O}(\varepsilon).
\end{equation}
Next, replacing $H(V_i)$ by $H(V_K)$ in the denominators of (\ref{eq::Degree0CaseIfNotUserK}) and (\ref{eq::Degree0CaseIfNotUserKNo2}) leaves, thanks to (\ref{9}), the corresponding right hand sides unchanged. Now, applying Lemma \ref{lemtao} with $X=bV_K$, $Y_1=aV_i$, and $Y_2= b \sum_{j \neq i,K}^{K} h_{ij}V_j$, yields
\begin{equation}
\frac{H \left(aV_i+b\sum_{j \neq i}^{K} h_{ij} V_j \right)}{H(V_K)} = 1+ \mathcal{O}(\varepsilon),
\end{equation}
which, after replacing $H(V_K)$ in the denominator by $H(V_i)$, again thanks to (\ref{9}), establishes (\ref{eq::Degree0CaseSignalNoise1}) for $i=1, \mydots,K-1$, as desired.
\par
We now proceed with the base case of induction over the maximum degree of the monomials $P$ and $Q$, namely $d=1$. Specifically, we need to show that for the specific $i$ in (\ref{eq::h_iiPQ}),   
\begin{equation} \label{eq::Degree1CaseSignalNoise1}
    \frac{H   \left( \frac{a}{b} \frac{h_{mn}}{h_{p\ell}} V_i +  \sum_{j \neq i}^{K} h_{ij} V_j \right)}{H(V_i)}=1+\mathcal{O}(\varepsilon),
\end{equation}
for all $a,b \in \mathbb{Z} \setminus \{0 \}$, and $h_{mn}, h_{p \ell}$ off-diagonal elements of $\mathbf{\tilde{H}}$, i.e., $m,n,p,\ell \in \{1,\mydots, K\}$, with $m \neq n$, $p \neq \ell$. We first consider the case $i=\ell=n=K$. As here $h_{mK}=h_{pK}=h_{iK}=1$, for $i=1, \mydots,K-1$,  (\ref{eq::Degree1CaseSignalNoise1}) becomes
\begin{equation} \label{eq::Degree1CaseSignalNoise1,i=l=n=K=1}
    \frac{H   \left( a  V_K+  b \sum_{j =1}^{K-1}  V_j \right)}{H(V_K)}=1+\mathcal{O}(\varepsilon), 
\end{equation}
and we are done thanks to (\ref{eq::Degree0CasePunchLineFinalForK}). Next, we consider $i=\ell=K$, $n \neq K$. In this case, (\ref{eq::Degree1CaseSignalNoise1}) reduces to
\begin{equation} \label{eq::Degree1CaseSignalNoise1,i=l=K,nneqK}
    \frac{H   \left( a h_{mn} V_K+  b \sum_{j =1}^{K-1}  V_j \right)}{H(V_K)}=1+\mathcal{O}(\varepsilon).
\end{equation}
We first note that, thanks to (\ref{8}) for $i=m$ and (\ref{wki2}), we have 
\begin{equation} \label{eq::Degree1,i=l=K,nneqK, ineq1}
    \frac{H   \left(a h_{mn} V_n+ a V_K \right)}{H(V_K)}=1+\mathcal{O}(\varepsilon).
\end{equation}
 Using (\ref{eq::Degree0CaseSumOfVi1Conclusion}), (\ref{wki2}), and (\ref{9}), we obtain 
 \begin{equation} \label{eq::Degree1,i=l=K,nneqK, ineq2}
  \frac{H   \left( a h_{mn} V_K+  a h_{mn} V_n \right)}{H(V_K)}=1+\mathcal{O}(\varepsilon).
\end{equation}
Replacing the denominators in (\ref{eq::Degree1,i=l=K,nneqK, ineq1}) and (\ref{eq::Degree1,i=l=K,nneqK, ineq2}) with $H(V_n)$ by using (\ref{9}), and applying Lemma \ref{lemtao} with $X=a h_{mn} V_n$, $Y_1=a\tilde{V}_K$, $Y_2= a h_{mn} V_K$, where $\tilde{V}_K$ is an independent copy of $V_K$, results in
\begin{equation} \label{eq::Degree1,i=l=K,nneqK, tao1}
   \frac{H   \left( a h_{mn} V_n+ a \tilde{V}_K+ a h_{mn} V_K   \right)}{H(V_n)}=1+\mathcal{O}(\varepsilon) .  
\end{equation}
Again using  (\ref{9}) to replace the denominator in (\ref{eq::Degree1,i=l=K,nneqK, tao1}) with $H(V_K)$, and applying (\ref{wki2}) yields
\begin{equation} \label{eq::Degree1,i=l=K,nneqK, tao1Clean}
   \frac{H   \left( a \tilde{V}_K+ a h_{mn} V_K   \right)}{H(V_K)}=1+\mathcal{O}(\varepsilon) .  
\end{equation}
We now combine (\ref{eq::Degree1,i=l=K,nneqK, tao1}) with (\ref{eq::Degree0CasePunchLineFinalForK}) (for $i=K$), to apply Lemma~\ref{lemtao} with $X= a \tilde{V}_K$, $Y_1=ah_{mn}V_K$, and $Y_2=b \sum_{j=1}^{K-1} V_j$, resulting in
\begin{equation} \label{eq::Degree1,i=l=K,nneqK, tao2}
   \frac{H   \left( a \tilde{V}_K+ a h_{mn} V_K + b \sum_{j=1}^{K-1} V_j    \right)}{H(V_K)}=1+\mathcal{O}(\varepsilon) ,  
\end{equation}
which, thanks to (\ref{wki2}), yields 
\begin{equation} \label{eq::Degree1,i=l=K,nneqK, tao2Clean}
   \frac{H   \left(a h_{mn} V_K + b \sum_{j=1}^{K-1} V_j    \right)}{H(V_K)}=1+\mathcal{O}(\varepsilon) ,  
\end{equation}
as desired. \par
Next, we consider $i=n=K$, $\ell \neq K$. In this case (\ref{eq::Degree1CaseSignalNoise1}) becomes
\begin{equation}
\label{eq::Degree1, i=n=K. lneqK}
    \frac{H   \left( a  V_K +  b h_{p\ell} \sum_{j =1}^{K-1}  V_j \right)}{H(V_i)}=1+\mathcal{O}(\varepsilon). 
\end{equation}
We first note that, thanks to (\ref{8}) for $i=p$ and (\ref{wki2}) for $i=p$, $j=\ell$, we have 
\begin{equation} \label{eq::Degree1,i=n=K. lneqK, ineq1}
    \frac{H   \left(a h_{p\ell} V_\ell+ a V_K \right)}{H(V_\ell)}=1+\mathcal{O}(\varepsilon).
\end{equation}
 Using (\ref{eq::Degree0CaseSumOfVi1LineK}) and (\ref{wki2}), we obtain 
 \begin{equation} \label{eq::Degree1,i=n=K. lneqK, ineq2}
  \frac{H   \left( a h_{p\ell} V_K+  a h_{p\ell} V_\ell \right)}{H(V_K)}=1+\mathcal{O}(\varepsilon).
\end{equation}
Replacing the denominator in  (\ref{eq::Degree1,i=n=K. lneqK, ineq2}) with $H(V_\ell)$ by using (\ref{9}), we apply Lemma \ref{lemtao} with $X=a h_{p\ell} V_\ell$, $Y_1=a V_K$, $Y_2= a h_{p\ell} \tilde{V} _K$, where $\tilde{V}_K$ is an independent copy of $V_K$, to obtain
\begin{equation} \label{eq::Degree1,i=n=K. lneqK, tao1}
   \frac{H   \left( a h_{p\ell} V_\ell+ a V_K+ a h_{p\ell} \tilde{V}_K   \right)}{H(V_\ell)}=1+\mathcal{O}(\varepsilon) .  
\end{equation}
We again use  (\ref{9}) to replace the denominator in (\ref{eq::Degree1,i=n=K. lneqK, tao1}) with $H(V_K)$ and get, thanks to (\ref{wki2}),
\begin{equation} \label{eq::Degree1,i=n=K. lneqK, tao1Clean}
   \frac{H   \left( a V_K+ a h_{p\ell} \tilde{V}_K   \right)}{H(V_K)}=1+\mathcal{O}(\varepsilon) .  
\end{equation}
We now combine (\ref{eq::Degree1,i=n=K. lneqK, tao1Clean}) with (\ref{eq::Degree0CasePunchLineFinalForK}) for $i=K$, to apply Lemma~\ref{lemtao} with $X= a h_{p\ell} \tilde{V}_K$, $Y_1=a V_K$, $Y_2=b h_{p\ell} \sum_{j=1}^{K-1} V_j$, and obtain
\begin{equation} \label{eq::Degree1,i=n=K. lneqK, tao2}
   \frac{H   \left( a h_{p\ell} \tilde{V}_K+ a V_K + b h_{p\ell} \sum_{j=1}^{K-1} V_j    \right)}{H(V_K)}=1+\mathcal{O}(\varepsilon) ,  
\end{equation}
which, thanks to (\ref{wki2}), yields 
\begin{equation} \label{eq::Degree1,i=n=K. lneqK, tao2Clean}
   \frac{H   \left(a  V_K + b h_{p\ell} \sum_{j=1}^{K-1} V_j    \right)}{H(V_K)}=1+\mathcal{O}(\varepsilon) ,  
\end{equation}
as desired. \par

We next consider $i=K, n \neq K, \ell \neq K$. Combining (\ref{eq::Degree1,i=l=K,nneqK, tao1Clean}) and (\ref{eq::Degree1,i=n=K. lneqK, tao2Clean}), we apply Lemma \ref{lemtao} with $X= a \tilde{V}_K$, $Y_1=a h_{mn} V_K$, and $Y_2= b h_{p\ell} \sum_{j=1}^{K-1} V_j$, where $\tilde{V}_K$ is an independent copy of $V_K$, to get
\begin{equation} \label{eq::Degree1,i=K, lneqK, nneqK, tao}
   \frac{H   \left(a  \tilde{V}_K +ah_{mn} V_K+ b h_{p\ell} \sum_{j=1}^{K-1} V_j    \right)}{H(V_K)}=1+\mathcal{O}(\varepsilon) ,  
\end{equation}    
which, thanks to (\ref{wki2}), yields 
\begin{equation} \label{eq::Degree1,i=K, lneqK, nneqK, taoClean}
   \frac{H   \left(a h_{mn}  V_K + b h_{p\ell} \sum_{j=1}^{K-1} V_j    \right)}{H(V_K)}=1+\mathcal{O}(\varepsilon) ,  
\end{equation}
as desired. \par
We finally consider $i \neq K$. Using (\ref{eq::Degree0CaseSumOfVi1Conclusion}), (\ref{9}) with $i=1$, $j=K$, and (\ref{wki2}), we get
\begin{equation}\label{eq::Degree1,nothingEqK, 1}
    \frac{H   \left( h_{mn}  V_i +  h_{mn} V_K    \right)}{H(V_K)}=1+\mathcal{O}(\varepsilon).
\end{equation}
Applying Lemma \ref{lemtao} with (\ref{eq::Degree1,i=l=K,nneqK, tao1Clean}) and (\ref{eq::Degree1,nothingEqK, 1}),  with $X= h_{mn} V_K $, $Y_1= \tilde{V}_K $, and $Y_2=  h_{mn} V_i$, yields
\begin{equation}\label{eq::Degree1,nothingEqK, tao1}
    \frac{H   \left( h_{mn}  V_i +  h_{mn} V_K +  \tilde{V}_K   \right)}{H(V_K)}=1+\mathcal{O}(\varepsilon),
\end{equation}
which, thanks to (\ref{wki2}), results in 
\begin{equation}\label{eq::Degree1,nothingEqK, tao1Clean}
    \frac{H   \left( h_{mn}  V_i +  \tilde{V}_K  \right)}{H(V_K)}=1+\mathcal{O}(\varepsilon).
\end{equation}
Again applying Lemma \ref{lemtao}, but now with  (\ref{eq::Degree1,i=n=K. lneqK, tao1Clean}) and (\ref{eq::Degree1,nothingEqK, tao1Clean}), with $X=\tilde{V}_K$, $Y_1=h_{p\ell} V_K$, and $Y_2=h_{mn}V_i$, yields
\begin{equation}\label{eq::Degree1,nothingEqK, tao2}
    \frac{H   \left( h_{p\ell} V_K + \tilde{V}_K+ h_{mn}V_i    \right)}{H(V_K)}=1+\mathcal{O}(\varepsilon),
\end{equation}
which, thanks to (\ref{wki2}), results in
\begin{equation}\label{eq::Degree1,nothingEqK, tao2Clean}
    \frac{H   \left( h_{p\ell} V_K + h_{mn}V_i    \right)}{H(V_K)}=1+\mathcal{O}(\varepsilon).
\end{equation}
Next, noting that $h_{iK}=1$, application of (\ref{8}) to user $i$ and (\ref{9}) to users $i$ and $K$ yields
\begin{equation}\label{eq::Degree1,nothingEqK, 3}
    \frac{H   \left( h_{p\ell} V_K + h_{p\ell} \sum_{j \neq i, K} h_{ij}V_j    \right)}{H(V_K)}=1+\mathcal{O}(\varepsilon).
\end{equation}
We now combine (\ref{eq::Degree1,nothingEqK, tao2Clean}) and (\ref{eq::Degree1,nothingEqK, 3}), and employ Lemma \ref{lemtao} with $X=h_{p\ell} V_K$, $Y_1=h_{mn}V_i$, $Y_2=h_{p\ell}\sum_{j \neq i,K} h_{ij}V_j $, to get 
\begin{equation}\label{eq::Degree1,nothingEqK, tao3}
    \frac{H   \left(  h_{mn}V_i+h_{p\ell}\sum_{j \neq i}^{K} h_{ij}V_j    \right)}{H(V_K)}=1+\mathcal{O}(\varepsilon).
\end{equation}
Applying \cite[Th. 14]{Wu15} (see Appendix \ref{appEntropyDiff}) with $p=a, q=b, X=h_{mn}V_i$, and $Y=h_{p\ell}\sum_{j \neq i}^{K} h_{ij} V_j$, and dividing the result thereof by $H(V_K)$ yields
 \begin{align}\label{eq::Degree1, notehingEqK, 4Thm14} 
\nonumber \frac{H \left(a h_{mn}V_i+b h_{p\ell}\sum_{j \neq i}^{K} h_{ij}V_j \right)}{H(V_K)}  -\frac{ H \left(h_{mn}V_i+h_{p\ell}\sum_{j \neq i}^{K} h_{ij}V_j \right)}{H(V_K)} \\ \leqslant   \tau_{a,b} \left (\frac{2 H \left(h_{mn}V_i+h_{p\ell}\sum_{j \neq i}^{K} h_{ij}V_j \right)-H(V_i)-H \left(\sum_{j \neq i}^{K} h_{ij} V_j \right)}{H(V_K)} \right ),
\end{align}
where $\tau_{a,b}= 7 \lfloor \log |a| \rfloor + 7 \lfloor \log |b| \rfloor +2$. The first, second, and third terms on the RHS of (\ref{eq::Degree1, notehingEqK, 4Thm14}) are $2+ \mathcal{O}(\varepsilon)$, $1+ \mathcal{O}(\varepsilon)$, and $1+ \mathcal{O}(\varepsilon)$, respectively, thanks to (\ref{eq::Degree1,nothingEqK, tao3}),  (\ref{9}), and (\ref{eq::Degree1,nothingEqK, 3}), respectively. Hence, the RHS of (\ref{eq::Degree1, notehingEqK, 4Thm14}) equals $\mathcal{O}(\varepsilon)$, resulting in
\begin{equation} \label{eq::Degree1, notehingEqK, 4Thm14Evolves} 
1 \leqslant \frac{H \left(a h_{mn}V_i+b h_{p\ell}\sum_{j \neq i}^{K} h_{ij}V_j \right)}{H(V_K)}  \leqslant 1+ \mathcal{O} (\varepsilon), 
\end{equation}
where the first inequality is due to (\ref{wki2}). Finally, using (\ref{9}) to replace the denominator of (\ref{eq::Degree1, notehingEqK, 4Thm14Evolves}) with $H(V_i)$, we obtain
\begin{equation} \label{eq::Degree1, notehingEqK, Final} 
\frac{H \left(a h_{mn}V_i+b h_{p\ell}\sum_{j \neq i}^{K} h_{ij}V_j \right)}{H(V_i)} = 1+ \mathcal{O} (\varepsilon), 
\end{equation}
as desired. \par
We proceed to the induction step over $d$ for (\ref{eq::InductionArgumentToProve}). To this end, we assume that (\ref{eq::InductionArgumentToProve}) holds for $d=m-1$, with $m \geqslant 2$, and we show that this implies (\ref{eq::InductionArgumentToProve}) for $d=m$. Consider $P:=a \prod_{j,k, j \neq k}^{K(K-1)}h_{jk}^{\alpha_{jk}}$ and $Q:=b \prod_{j,k, j \neq k}^{K(K-1)}h_{jk}^{\beta_{jk}}$ such that the maximum of the degrees of $P$ and $Q$ is $d$. Next, note that we can factorize $P$ and $Q$ such that $P=P_1 P_2$ and $Q=Q_1 Q_2$, where $P_1, P_2, Q_1,$ and $Q_2$ are all of degree strictly smaller than $d$. We now want to show that
\begin{equation*}
\frac{H \left( \frac{P_1 P_2}{Q_1 Q_2} V_i + \displaystyle \sum_{j \neq i} h_{ij} V_j \right)}{H(V_i)} = 1+ \mathcal{O} (\varepsilon).
\end{equation*}
To this end, first note that thanks to the induction assumption, we have 
\begin{equation} \label{eq::degreeKeq1}
 \frac{H \left( \frac{P_1 }{Q_1} V_i + \displaystyle\sum_{j \neq i} h_{ij} V_j \right)}{H(V_i)} = 1+ \mathcal{O} (\varepsilon),   
\end{equation}
and
\begin{equation} \label{eq::degreeKeq2}
 \frac{H \left( \frac{P_2}{Q_2} V_i + \displaystyle\sum_{j \neq i} h_{ij} V_j \right)}{H(V_i)} = 1+ \mathcal{O} (\varepsilon).   
\end{equation}
Next, we use (\ref{eq::Degree0CaseSignalNoise1}) with $a=b=1$, as well as  (\ref{eq::degreeKeq1}) and (\ref{eq::degreeKeq2}), upon replacing their denominators with $H\left(\sum_{j \neq i} h_{ij} V_j\right)$, which is possible thanks to (\ref{8}), and we apply Lemma \ref{lemtao} with $X=\sum_{j \neq i} h_{ij} V_j$, $Y_1=\frac{P_1 }{Q_1} V_i$, $Y_2=\frac{P_2 }{Q_2} V_i$, $Y_3= V_i$, to get
\begin{equation}  \label{eq::degreeKTao}
 \frac{H \left(V_i+ \frac{P_1}{Q_1} V_i + \frac{P_2}{Q_2} V_i + \displaystyle\sum_{j \neq i} h_{ij} V_j \right)}{H\left(\displaystyle\sum_{j \neq i} h_{ij} V_j \right)} = 1+ \mathcal{O} (\varepsilon).
\end{equation}
Thanks to (\ref{wki2}) this results in 
\begin{equation} \label{eq::degreeKTaoClean}
 \frac{H \left(\frac{P_2}{Q_2}V_i+ \frac{P_1}{Q_1} \frac{P_2}{Q_2} V_i   \right)}{H\left(\displaystyle \sum_{j \neq i} h_{ij} V_j \right)} = 1+ \mathcal{O} (\varepsilon).     
\end{equation}
Finally, we use (\ref{eq::degreeKeq2}) and (\ref{eq::degreeKTaoClean}) upon replacing its denominator with $H(V_i)$, which is possible thanks to (\ref{8}), and we apply Lemma \ref{lemtao} with $X=\frac{P_2}{Q_2}V_i$, $Y_1= \sum_{j \neq i} h_{ij} V_j$, $Y_2= \frac{P_1 P_2}{Q_1 Q_2}V_i$, to conclude that
\begin{equation} \label{eq::degreeKTao2}
 \frac{H \left(\frac{P_2}{Q_2}V_i+ \frac{P_1}{Q_1} \frac{P_2}{Q_2} V_i  +\displaystyle \sum_{j \neq i} h_{ij} V_j  \right)}{H\left( V_i \right)} = 1+ \mathcal{O} (\varepsilon).
\end{equation}
Again, thanks to (\ref{wki2}), this yields 
\begin{equation} \label{eq::degreeKTao2Final}
 \frac{H \left( \frac{P_1}{Q_1} \frac{P_2}{Q_2} V_i  +\displaystyle \sum_{j \neq i} h_{ij} V_j  \right)}{H\left( V_i \right)} = 1+ \mathcal{O} (\varepsilon),    
\end{equation}
as desired and thereby concludes the induction step with respect to the maximum degree, which, in turn, establishes the base case for the induction over the number of terms. \par
We proceed to the induction over the number of terms by assuming that  (\ref{eq::InductionArgumentToProve}) holds for $p=n$, $ n \geqslant 1$. Recall that $p$ is the maximum number of terms in the polynomials $P:= \sum_{s=1}^{\ell} P_s$ and $Q := \sum_{s=1}^{m} Q_s$, where $P_s=a_s \prod_{j,k, j \neq k}^{K(K-1)}h_{jk}^{\alpha^{(s)}_{jk}}$ and $Q_s=b_s \prod_{j,k, j \neq k}^{K(K-1)}h_{jk}^{\beta^{(s)}_{jk}}$, with $a_s,b_s \in \mathbb{Z}$ and $\alpha^{(s)}_{jk},\beta^{(s)}_{jk} \in \mathbb{N}$.
We need to show that this implies 
\begin{equation} \label{eq::p==n,Induction}
 \frac{H \left( \frac{ \left(P_1 + P_2 + \, \mydots \, +  P_\ell \right)}{\left(Q_1 + Q_2 + \, \mydots \, +Q_m \right)}  V_i  +\displaystyle \sum_{j \neq i} h_{ij} V_j  \right)}{H\left( V_i \right)} = 1+ \mathcal{O} (\varepsilon),   
\end{equation}
for $\max \{\ell, m \}= n+1$. 
First, we consider the case $\ell = n+1 > m$. Then, thanks to 
the induction assumption, it holds that
\begin{equation} \label{eq::p==n,Induction2}
 \frac{H \left(  P_1 V_i  +\displaystyle (Q_1 + \mydots + Q_m) \sum_{j \neq i} h_{ij} V_j  \right)}{H\left( V_i \right)} = 1+ \mathcal{O} (\varepsilon),
\end{equation}
and 
\begin{equation} \label{eq::p==n,Induction3}
 \frac{H \left( \left (P_2 + \mydots + P_\ell \right) V_i  +\displaystyle (Q_1 + \mydots + Q_m) \sum_{j \neq i} h_{ij} V_j  \right)}{H\left( V_i \right)} = 1+ \mathcal{O} (\varepsilon).
\end{equation}

Using (\ref{eq::p==n,Induction2}) and (\ref{eq::p==n,Induction3}), we apply Lemma \ref{lemtao} with $X= (Q_1 + \mydots + Q_m) \left ( \sum_{j \neq i} h_{ij} V_j  \right)$, $Y_1= (P_2 + \mydots + P_\ell) V_i$, and $Y_2=P_1 V_i$, to get
\begin{equation} \label{eq::p==n,InductionProved1}
 \frac{H \left( \left(P_1 + P_2 + \mydots+ P_{\ell} \right) V_i  +\displaystyle (Q_1 + \mydots + Q_m) \sum_{j \neq i} h_{ij} V_j  \right)}{H\left( \sum_{j \neq i} h_{ij} V_j \right)} = 1+ \mathcal{O} (\varepsilon),
\end{equation}
which, upon replacing the denominator with $H(V_i)$, which is possible thanks to (\ref{8}), establishes (\ref{eq::p==n,Induction}) as desired. \par
We next consider the case $m=n+1 >\ell$. In this case, similarly, we apply Lemma \ref{lemtao} with $X=P_1 V_i$, $Y_1= Q_1 \sum_{j \neq i} h_{ij} V_j$, and $Y_2 =(Q_2 + \mydots + Q_m ) \sum_{j \neq i} h_{ij} V_j$, to obtain  (\ref{eq::p==n,Induction}) as desired. \par

We are left with the case $m=\ell=n+1$. First, note that we have already shown (\ref{eq::p==n,Induction}) for $\ell=n+1 >m$, and hence the following relations hold: 
\begin{equation} \label{eq::p==n,Induction4}
 \frac{H \left(  (P_1+ \mydots + P_\ell) V_i  +\displaystyle (Q_1 + \mydots + Q_{m-1}) \sum_{j \neq i} h_{ij} V_j  \right)}{H\left( V_i \right)} = 1+ \mathcal{O} (\varepsilon).
\end{equation}
and
 \begin{equation} \label{eq::p==n,Induction5}
 \frac{H \left(  (P_1 + \mydots +P_\ell) V_i  +\displaystyle Q_m \sum_{j \neq i} h_{ij} V_j  \right)}{H\left( V_i \right)} = 1+ \mathcal{O} (\varepsilon).
\end{equation}
We now combine (\ref{eq::p==n,Induction4}) and (\ref{eq::p==n,Induction5}) and apply Lemma \ref{lemtao} with $X= (P_1 + \mydots +P_{\ell}) V_i $, $Y_1=  (Q_1 + \mydots + Q_{m-1}) \sum_{j \neq i} h_{ij} V_j$, and $Y_2= Q_m \sum_{j \neq i} h_{ij} V_j$ to obtain  (\ref{eq::p==n,Induction}) as desired. \par

It remains to prove necessity for the non-fully-connected case. We start with the following technical result. 

\begin{lemma} \label{lemma::setMembershipNFC}
 Let $\mathcal{T}$ be the set of $K$-user non-fully-connected IC matrices of an arbitrary, but fixed topology. Let $\mathcal{T'}$ be the subset of $\mathcal{T}$ obtained by restricting the set of fully-connected matrices that are in $\mathcal{S}$ according to Proposition \ref{prop} for $\text{DoF}_i=1/2$, $i=1, \mydots, K$, to the complement of the zero-set of $\mathcal{T}$. Then, the set  $\mathcal{T'}$ is an a.a. subset of $\mathcal{T}$.

 \begin{proof}
  Since the set of fully-connected matrices $\mathcal{F}$ and the set $\mathcal{S}$ in Proposition \ref{prop} for $\text{DoF}_i=1/2$, $i=1, \mydots, K$, are  almost all subsets of $\mathbb{R}^{K \times K}$, the set $\mathcal{F}_s := \mathcal{F} \cap \mathcal{S}$, as the intersection of two almost all sets, is also an almost all subset of $\mathbb{R}^{K \times K}$. Next, note that $\mathcal{T}$ is obtained by restricting $\mathcal{F}$ to the complement of the zero-set of $\mathcal{T}$. Now, let us assume, by way of contradiction, that $\mathcal{T} \setminus \mathcal{T'}$ has positive measure. As 
  $\mathcal{T} \setminus \mathcal{T'}$ is the restriction of $\mathcal{F} \setminus \mathcal{F}_s$ to the complement of the zero-set of $\mathcal{T}$, if $\mathcal{T} \setminus \mathcal{T'}$ were of positive measure so would $\mathcal{F} \setminus \mathcal{F}_s$ have to be, which constitutes a contradiction and thereby finishes the proof. 
 \end{proof}
\end{lemma}

 We finalize the proof by contradiction. Let $\mathbf{H}$ be in the almost all set $\mathcal{T'}$ defined in Lemma \ref{lemma::setMembershipNFC} and assume that $\mathbf{H}$ and all scaled versions thereof violate Condition~($\ast$), while each user achieves $1/2$ DoF. Let $\mathbf{\tilde{H}}$ be a fully-connected IC matrix in $\mathcal{S}$ for $\text{DoF}_i=1/2$, $i=1, \mydots, K$, which, upon restriction to the complement of the zero-set of $\mathcal{T}$ corresponding to $\mathcal{T'}$, yields $\mathbf{H}$.
 Next, we observe that if injectivity is violated for a set, it is also violated for its supersets. Hence, since $\mathcal{W}^{(\mathbf{H})} \subseteq \mathcal{W}^{(\mathbf{\tilde{H}})}$, and all diagonal entries of $\mathbf{\tilde{H}}$ and $\mathbf{H}$ are identical, if $\mathbf{H}$ and all scaled versions thereof violate Condition ($\ast$), so do $\mathbf{\tilde{H}}$ and all scaled versions thereof. But then, however, as necessity was already established above for all fully-connected IC matrices in the set $\mathcal{S}$ for $\text{DoF}_i=1/2$, $i=1, \mydots, K$, we are left with a contradiction.
 


\section{Proof of Necessity in Theorem~\ref{thm1} for all 3-User Non-Fully-Connected IC Matrices} \label{section::3UserNonFullyConnected}

While Theorem~\ref{thm1} established necessity for almost all IC matrices and for arbitrary $K$, a stronger result is possible in (at least) the $3$-user case. Specifically, necessity in Theorem \ref{thm1} can be shown to hold for \textit{all} channel matrices $\mathbf{H}$ of every fixed non-fully-connected topology. The corresponding proof proceeds by direct enumeration of all possible channel topologies combined with the application of Condition~($\ast$) and a result in \cite{Wu15}.





For the first two topologies necessity follows directly as the topology per se implies the existence of a scaled version of $\mathbf{H}$ for which Condition~($\ast$) holds. (Note that scaling does not change the topology.)

\textit{Topology 1}. We consider the case $h_{ij} = h_{ji}=0$, for $i \neq j$, and set w.l.o.g.\footnote {This specific choice comes w.l.o.g, as we can simply relabel the users, e.g. when $h_{23}=h_{32}=0$, we relabel user 1 as user 2 and user 2 as user 3. We will exploit this symmetry in all topologies.}  $i=1$ and $j=2$, so that 
 \begin{equation*} \mathbf{H}=\begin{pmatrix} h_{11}& 0 &h_{13}\\ 0&h_{22}&h_{23}\\ h_{31}&h_{32}&h_{33} \end{pmatrix}\hspace{-2.6pt}.  \end{equation*}
The remaining links (apart from the direct links between the users corresponding to the diagonal entries) may or may not be present, i.e., $h_{13}, h_{23}, h_{31}$, $h_{32}$ may or may not be nonzero. 
We now scale  $\mathbf{H}$ to convert it into 
\begin{equation} \label{tildeHNotFully}  \mathbf{\tilde{H}}=\begin{pmatrix} \sqrt{2}&0&\tilde{h}_{13}\\ 0&g_2&\tilde{h}_{23}\\ \tilde{h}_{31}&\tilde{h}_{32}&g_3 \end{pmatrix}\hspace{-2.6pt}, \end{equation}
 where $\tilde{h}_{13}, \tilde{h}_{23}, \tilde{h}_{31}, \tilde{h}_{32} \in \{ 0,1 \}$, and $g_2, g_3 \in \mathbb{R} \setminus \{0\}$. As $\sqrt{2}$ is irrational and $\mathcal{W}^{(\mathbf{\tilde{H}})} \subseteq \mathbb{N}$, (\ref{alternativeFormulationOfConditionStar}) implies that user~1 cannot violate Condition~($\ast$) in $\mathbf{\tilde{H}}$. 
 If user~2 is to violate Condition~($\ast$) in $\mathbf{\tilde{H}}$, then $g_2$ must be in $\mathbb{Q}$. Assuming that this is, indeed, the case, we next scale the first and second rows of $\tilde{\mathbf{H}}$ by $\frac{\sqrt{3}}{g_2}$ and the third column by $\frac{g_2}{\sqrt{3}}$ (to keep the off-diagonal components in $\{ 0, 1 \}$)  to get
\begin{equation*}  \mathbf{\overline{H}}=\begin{pmatrix} \frac{\sqrt{6}}{g_2}&0& \tilde{h}_{13}\\ 0&\sqrt{3}&\tilde{h}_{23}\\ \tilde{h}_{31}&\tilde{h}_{32}&\frac{g_2g_3}{\sqrt{3}} \end{pmatrix}\hspace{-2.6pt}. \end{equation*}
Owing to $g_2 \in \mathbb{Q}$, the first diagonal entry in $\mathbf{\overline{H}}$ is irrational so that user 1 cannot violate Condition~($\ast$) in $\mathbf{\overline{H}}$. Likewise, user 2 cannot violate Condition~($\ast$)  in $\mathbf{\overline{H}}$ as $\sqrt{3}$ is irrational. If user 3 is to violate Condition~($\ast$)  in $\mathbf{\overline{H}}$, then $\frac{g_3}{\sqrt{3}}$ must be in $\mathbb{Q}$. Assuming that this is, indeed, the case, we next
scale the third row of $\mathbf{\overline{H}}$ by $\frac{\sqrt{15}}{g_2g_3}$, and the first and second column by $\frac{g_2g_3}{\sqrt{15}}$ to obtain 
\begin{equation*}  \mathbf{H'}=\begin{pmatrix} \frac{g_3 \sqrt{2}}{\sqrt{5}}&0&\tilde{h}_{13}\\ 0& \frac{g_2g_3}{\sqrt{5}}&\tilde{h}_{23}\\\tilde{h}_{31}&\tilde{h}_{32}& \sqrt{5} \end{pmatrix}\hspace{-2.6pt}. \end{equation*}
Note that since $g_2, \frac{g_3}{\sqrt{3}} \in \mathbb{Q}$, all diagonal entries of $\mathbf{H'}$ are irrational, so that none of the users in $\mathbf{H'}$ can violate Condition~($\ast$). We have hence established that user $2$ violating Condition~($\ast$) in $\mathbf{\tilde{H}}$ implies the existence of a scaled version of $\mathbf{H}$, namely $\mathbf{H'}$, that satisfies Condition~($\ast$). It remains to treat the case of user 3 violating Condition~($\ast$) in $\mathbf{\tilde{H}}$. In that case, again, $g_3$ must be in $\mathbb{Q}$. We scale $\tilde{\mathbf{H}}$ to turn it into $\mathbf{H'}$, and note that user 1 and user 3 cannot violate Condition~($\ast$) in $\mathbf{H'}$ as $\frac{g_3\sqrt{2}}{\sqrt{5}}$ and $\sqrt{5}$ are irrational. If user 2 is to violate Condition~($\ast$) in $\mathbf{H'}$, as $g_3 \in \mathbb{Q}$, $\frac{g_2}{\sqrt{5}}$ must be in $\mathbb{Q}$, which, in turn, would result in an $\mathbf{\overline{H}}$ that has all its diagonal entries irrational and would therefore satisfy Condition~($\ast$). This establishes that user $3$ violating Condition~($\ast$) in $\mathbf{\tilde{H}}$ implies the existence of a scaled version of $\mathbf{H}$, namely $\mathbf{\overline{H}}$, that satisfies Condition~($\ast$). In summary, we have established that 
for every $\mathbf{H}$ of Topology 1 there always exists at least one scaled version of $\mathbf{H}$ satisfying Condition~($\ast$).

\textit{Topology 2}. We consider the case $h_{ik}=h_{ji}=h_{kj}=0$, for distinct $i, j$, $k$. For concreteness and again w.l.o.g, we set $i=1, j=2, k=3$, which leads to the following IC matrix
 \begin{equation*} \mathbf{H}=\begin{pmatrix} h_{11}&h_{12}&0\\ 0&h_{22}&h_{23}\\ h_{31}&0&h_{33} \end{pmatrix},\hspace{-2.6pt}\end{equation*}
with $h_{12}, h_{23}, h_{31} \neq 0$. Note that if any of $h_{12},h_{23},h_{31}$ were equal to zero, we would be back to Topology 1. 

We next scale $\mathbf{H}$ to convert it into  
\begin{equation*}  \mathbf{\tilde{H}}=\begin{pmatrix} \sqrt{2}&1&0\\ 0&\sqrt{2}&1\\1&0&g_3 \end{pmatrix}\hspace{-2.6pt}, \end{equation*}
where $g_3 \in \mathbb{R} \setminus \{0 \}$. As $\sqrt{2}$ is irrational and $\mathcal{W}^{(\mathbf{\tilde{H}})} \subseteq \mathbb{N}$ users 1 and 2 cannot violate Condition~($\ast$) in $\mathbf{\tilde{H}}$. If user 3 were to violate Condition~($\ast$) in $\mathbf{\tilde{H}}$, $g_3$ would have to be in $\mathbb{Q}$. In this case, we could convert $\mathbf{\tilde{H}}$ (by multiplying the first column and the third row by $\frac{g_3}{\sqrt{3}}$ and $\frac{\sqrt{3}}{g_3}$, respectively) into 
\begin{equation*}  \mathbf{\overline{H}}=\begin{pmatrix} \frac{g_3 \sqrt{2}}{\sqrt{3}}&1&0\\ 0&\sqrt{2}&1\\1&0&\sqrt{3} \end{pmatrix}\hspace{-2.6pt}, \end{equation*}
which can not violate Condition~($\ast$) as all its diagonal entries are irrational. Following the same playbook as in Topology 1, we have hence established that for every $\mathbf{H}$ of Topology 2, there always exists at least one scaled version of $\mathbf{H}$ satisfying Condition~($\ast$). \par

The proof for the remaining topologies is organized according to the number of missing links and we shall argue by contradiction in all cases as follows: Suppose that each user achieves $1/2$ DoF and Condition~($\ast$) is violated for $\mathbf{H}$ and all its scaled versions. Under these assumptions, we shall identify a scaled version of $\mathbf{H}$ that does not allow $3/2$ DoF in total, implying that, owing to Remark~\ref{remarkDoFScaling},  $\mathbf{H}$ itself does not allow $3/2$ DoF in total, which establishes the contradiction.

\textit{One missing link.} We start with the case where exactly one non-diagonal entry of the IC matrix is zero, i.e., $h_{ij}=0$, for $i \neq j$. For concreteness and w.l.o.g.\footnote{Again, the other choices for $i$ and $j$ follow by simply relabeling users.}, we set $i=1$ and $j=2$ and consider the corresponding IC matrix 
 \begin{equation*} \mathbf{H}=\begin{pmatrix} h_{11}&0&h_{13}\\ h_{21}&h_{22}&h_{23}\\ h_{31}&h_{32}&h_{33} \end{pmatrix}\hspace{-2.6pt},  \end{equation*}
with all coefficients but $h_{12}$ nonzero. 
We next scale $\mathbf{H}$ to convert it into 
\begin{equation*}  \mathbf{\tilde{H}}=\begin{pmatrix} g_1&0&1\\ 1&g_2&1\\1&1&g_3 \end{pmatrix}\hspace{-2.6pt}. \end{equation*}
As $\mathbf{\tilde{H}}$ violates Condition~($\ast$) by assumption and $\mathcal{W}^{(\mathbf{\tilde{H}})} \subseteq \mathbb{N}$, at least one diagonal entry of $\tilde{\mathbf{H}}$ must be in $\mathbb{Q}$. We can now apply \cite[Th. 8]{Wu15} (See Appendix \ref{4Th8restate}) as follows. If $g_1$ is in $\mathbb{Q}$, we set $i=1, j=3, k=2$ in \cite[Th. 8]{Wu15} to conclude that $\mathbf{\tilde{H}}$ does not allow $3/2$ DoF in total, thereby establishing the contradiction. The argument for $g_2$ or $g_3$ rational follows along the exact same lines. 

\textit{Two missing links.} Next, we consider the case where exactly two off-diagonal entries of the IC matrix are equal to zero. This case will be dealt with by splitting it up into five topologies as follows: $h_{ij}=h_{ji}=0$, $h_{ij}=h_{ik}=0$, $h_{ij}=h_{jk}=0$, $h_{ij}=h_{ki}=0$, and $h_{ij}=h_{kj}=0$, for distinct $i, j, k$. The first case is already covered by Topology 1. The remaining cases are organized into Topologies 3, 4, 5, and 6, respectively.    \newline
\textit{Topology 3.} We have $h_{ij}=h_{ik}=0$, for distinct $i, j, k$. For concreteness and again w.l.o.g., we set $i=1$, $j=2$, $k=3$, and consider the corresponding IC matrix

 \begin{equation*} \mathbf{H}=\begin{pmatrix} h_{11}&0& 0\\ h_{21}&h_{22}&h_{23}\\ h_{31}&h_{32}&h_{33} \end{pmatrix}\hspace{-2.6pt},  \end{equation*}
with $h_{21}, h_{23}, h_{31}, h_{32}$ all nonzero real numbers. 
Next, we scale $\mathbf{H}$ to convert it into  
\begin{equation*}  \mathbf{\tilde{H}}=\begin{pmatrix} \sqrt{2}&0&0\\ 1&g_2&1\\1&1&g_3 \end{pmatrix}\hspace{-2.6pt}, \end{equation*}
where $g_2$ and $g_3$ are nonzero real numbers. Now, note that user 1 in $\mathbf{\tilde{H}}$ cannot violate Condition~($\ast$) as $\sqrt{2}$ is irrational and $\mathcal{W}^{(\mathbf{\tilde{H}})} \subseteq \mathbb{N}$. If user 2 or user 3 were to violate Condition~($\ast$) in $\mathbf{\tilde{H}}$, then $g_2$ or $g_3$, respectively, would have to be in $\mathbb{Q}$, in which case we can again apply \cite[Th. 8]{Wu15}, namely with $i=2, j =1, k=3$ and $i=3, j=1, k=2$, respectively, to conclude that the total number of DoF in $\mathbf{\tilde{H}}$ is less than $3/2$. This establishes the contradiction.

\textit{Topology 4}. We have $h_{ij}=h_{jk}=0$, for distinct  $i, j, k$. For concreteness and again w.l.o.g., we set $i=1$, $j=2$, $k=3$, and consider the corresponding IC matrix
 \begin{equation*} \mathbf{H}=\begin{pmatrix} h_{11}&0& h_{13}\\ h_{21}&h_{22}&0\\ h_{31}&h_{32}&h_{33} \end{pmatrix}\hspace{-2.6pt},  \end{equation*}
 where $h_{13}, h_{21}, h_{31}$, and $h_{32}$ are nonzero real numbers.
Next, we scale $\mathbf{H}$ to convert it into  
\begin{equation*}  \mathbf{\tilde{H}}=\begin{pmatrix} \sqrt{2}&0&1\\ 1&g_2&0\\1&1&g_3 \end{pmatrix}\hspace{-2.6pt}, \end{equation*}
where $g_2$ and $g_3$ are nonzero real numbers. First, note that user 1 cannot violate Condition~($\ast$) in $\mathbf{\tilde{H}}$ as $\sqrt{2}$ is irrational and $\mathcal{W}^{(\mathbf{\tilde{H}})} \subseteq \mathbb{N}$. If user 2 were to violate Condition~($\ast$), $g_2$ would have to be in $\mathbb{Q}$, and we can apply  \cite[Th. 8]{Wu15} with $i=2, j =1, k=3$ to conclude that the total number of DoF is less than $3/2$, which establishes the contradiction. If user 3 were to violate Condition~($\ast$) in $\mathbf{\tilde{H}}$, $g_3$ would have to be in $\mathbb{Q}$. We then scale the first row and the third column of $\tilde{\mathbf{H}}$ by $\frac{g_3}{\sqrt{3}}$ and $\frac{\sqrt{3}}{g_3}$, respectively, to obtain 
\begin{equation*}  \mathbf{H'}=\begin{pmatrix} \frac{ \sqrt{2} g_3}{\sqrt{3}}&0&1\\ 1&g_2&0\\1&1&\sqrt{3} \end{pmatrix}\hspace{-2.6pt}. \end{equation*}
Now, we note that since $g_3 \in \mathbb{Q}$, the first diagonal entry in $\mathbf{H}'$ is irrational and hence user 1 cannot violate Condition~($\ast$) in $\mathbf{H'}$. User 3 cannot violate Condition~($\ast$) in $\mathbf{H}'$ as $\sqrt{3}$ is irrational. If user 2 in $\mathbf{H'}$ were to violate Condition~($\ast$), $g_2$ would have to be in $\mathbb{Q}$, and we can apply \cite[Th. 8]{Wu15} with $i=2, j =1, k=3$ to conclude that the total number of DoF is less than $3/2$, which establishes the contradiction.

\textit{Topology 5}. We consider the case $h_{ij}=h_{ki}=0$, for distinct $i, j, k$. For concreteness and w.l.o.g., we set $i=1$, $j=2$, $k=3$ and consider the corresponding IC matrix
 \begin{equation*} \mathbf{H}=\begin{pmatrix} h_{11}&0& h_{13}\\ h_{21}&h_{22}&h_{23}\\ 0&h_{32}&h_{33} \end{pmatrix}\hspace{-2.6pt},  \end{equation*}
with $h_{13}, h_{21}, h_{23}, h_{32}$ nonzero real numbers.
We next scale $\mathbf{H}$ to convert it into  
\begin{equation*}  \mathbf{\tilde{H}}=\begin{pmatrix} g_1&0&1\\ 1&g_2&1\\ 0&1& \sqrt{2} \end{pmatrix}\hspace{-2.6pt}, \end{equation*}
where $g_1$ and $g_2$ are nonzero real numbers. First, note that user 3 cannot violate Condition~($\ast$) in $\mathbf{\tilde{H}}$ as $\sqrt{2}$ is irrational and $\mathcal{W}^{(\mathbf{\tilde{H}})} \subseteq \mathbb{N}$. If user 2 were to violate Condition~($\ast$) in $\mathbf{\tilde{H}}$, $g_2$ would have to be in $\mathbb{Q}$. We then scale the third row and the second column of $\tilde{\mathbf{H}}$ by $\frac{g_2}{\sqrt{3}}$ and $\frac{\sqrt{3}}{g_2}$, respectively, to obtain 
\begin{equation*}  \mathbf{H'}=\begin{pmatrix}g_1&0&1\\ 1&\sqrt{3}&1\\0&1& \frac{ \sqrt{2} g_2}{\sqrt{3}} \end{pmatrix}\hspace{-2.6pt}, \end{equation*}
and note that since $g_2 \in \mathbb{Q}$, the second and the third diagonal entries of $\mathbf{H}'$ are irrational, which implies that 
users 2 and 3 cannot violate Condition~($\ast$) in $\mathbf{H'}$. If user 1 in $\mathbf{H}'$ were to violate Condition~($\ast$), $g_1$ would have to be in $\mathbb{Q}$ and we can apply \cite[Th. 8]{Wu15} with $i=1, j =3, k=2$ to conclude that the total number of DoF of $\mathbf{H}'$ is less than $3/2$. This establishes the contradiction.

\textit{Topology 6.} We finally consider the case $h_{ij}=h_{kj}=0$, and set w.l.o.g. $j=1$, $i=2$, $k=3$ to get  the IC matrix 
 \begin{equation*} \mathbf{H}=\begin{pmatrix} h_{11}& h_{12} &h_{13}\\ 0&h_{22}&h_{23}\\ 0&h_{32}&h_{33} \end{pmatrix}\hspace{-2.6pt},  \end{equation*}
with $h_{12}, h_{13}, h_{23}, h_{32}$ nonzero real numbers. We scale $\mathbf{H}$ to convert it into  
\begin{equation*}  \mathbf{\tilde{H}}=\begin{pmatrix} \sqrt{2}&1&1\\ 0&g_2&1\\0&1&g_3 \end{pmatrix}\hspace{-2.6pt},
 \end{equation*}
 where  $g_2$ and $g_3$ are nonzero real numbers. First, note that user 1 cannot violate Condition~($\ast$) in $\mathbf{\tilde{H}}$ as $\sqrt{2}$ is irrational and $\mathcal{W}^{(\mathbf{\tilde{H}})} \subseteq \mathbb{N}$. If user 2 or user 3 in $\mathbf{\tilde{H}}$ were to violate Condition~($\ast$), then $g_2$ or $g_3$, respectively, would have to be in $\mathbb{Q}$, and we can apply  \cite[Th. 8]{Wu15} with $i=2, j =3, k=1$ or $i=3, j=2, k=1$, respectively, to conclude that $3/2$ DoF in total cannot be achieved, thereby establishing the contradiction.

\textit{3 missing links.} We now consider IC matrices with exactly three off-diagonal entries equal to zero and enumerate the corresponding possible topologies as follows. First, we set, w.l.o.g., $h_{ij}=0$, for distinct $i,j$. We need to choose two more zero entries among the remaining off-diagonal coefficients $h_{ji}, h_{jk}, h_{ik}, h_{ki}, h_{kj}$, with $k \neq i,j$. There is a total of $\binom{5}{2} = 10$ choices. All choices including $h_{ji}=0$ result in Topology~1 and have hence already been dealt with. This leaves us with the $\binom{4}{2}=6$ choices $h_{ik}=h_{ki}=0$, $h_{jk}=h_{kj}=0$, $h_{jk}=h_{ki}=0$, $h_{ki}=h_{kj}=0$,  $h_{ik}=h_{kj}=0$, $h_{jk}=h_{ik}=0$. The first two cases are covered by Topology~1 and the third case is comprised by Topology~2. The remaining three topologies are identical. To see this, consider $h_{ij}=h_{jk}=h_{ik}=0$ and relabel the users according to $j'=k, k'=j$. This leads to $h_{ik'}=h_{k'j'}=h_{ij'}=0$, which is the fifth topology above.  Similarly, if we relabel the users according to $k'=i, i'=j, j'=k$, we obtain $h_{k'i'}=h_{i'j'}=h_{k'j'}=0$, which results in the fourth topology. The remaining case is dealt with by setting w.l.o.g. $i=1, j=2, k=3$ in $h_{ij}=h_{jk}=h_{ik}=0$, leading to 
 \begin{equation*} \mathbf{H}=\begin{pmatrix} h_{11}& 0 &0 \\ h_{21}&h_{22}&0\\ h_{31}&h_{32}&h_{33} \end{pmatrix}\hspace{-2.6pt},  \end{equation*}
with $h_{21}, h_{31}$, $h_{32}$ nonzero real numbers. We scale $\mathbf{H}$ to convert it into
\begin{equation*}  \mathbf{\tilde{H}}=\begin{pmatrix} \sqrt{2}&0&0\\ 1&g_2&0\\1&1& \sqrt{2} \end{pmatrix}\hspace{-2.6pt}, \end{equation*}
where $g_2$ is a nonzero real number, and note that users 1 and 3 cannot violate Condition~($\ast$) in $\mathbf{\tilde{H}}$ as  $\sqrt{2}$ is irrational and $\mathcal{W}^{(\mathbf{\tilde{H}})} \subseteq \mathbb{N}$. If user 2 in $\mathbf{\tilde{H}}$ is to violate Condition~($\ast$), then $g_2$ must be in $\mathbb{Q}$, and we can apply \cite[Th. 8]{Wu15} with $i=2, j =1, k=3$ to conclude that $3/2$ DoF in total cannot be achieved, thereby establishing the desired contradiction. \par

\textit{More than 3 missing links}. For IC matrices with more than three off-diagonal entries equal to zero, there always exist users $i, j$ such that $h_{ij}=h_{ji}=0$ and hence we are back to Topology~1. 

\section{AN APPLICATION} \label{sec:applications}
We now show how our results allow to develop a significant generalization of \cite[Thm. 8]{Wu15}, which was the main technical engine in our proof of necessity for \textit{all} $3$-user channel matrices in the previous section. Specifically, we provide an extension of \cite[Thm. 8]{Wu15} from the 3-user case to the $K$-user case, which, in addition, applies to almost all channel matrices whereas \cite[Thm. 8]{Wu15} applies to the measure-zero set of channel matrices with algebraic off-diagonal entries only. We are also able to relax the assumption of the channel coefficients $h_{ii}, h_{ij}, h_{ki},h_{kj}$ in \cite[Thm. 8]{Wu15} being nonzero rational numbers to allow a.a. real numbers. Finally, \cite[Thm. 8]{Wu15} makes a statement on the total number of DoF, whereas our extension is in terms of DoF achievable by individual users. 
\begin{theorem} \label{appTheoremGen}
For almost all $K \times K$ IC matrices $\mathbf{H}$, if there exist distinct users $i, j, k$ such that $\frac{h_{ii}h_{kj}}{h_{ij}h_{ki}}$ is a non-zero rational number, then 1/2 DoF for each user cannot be achieved. 
\begin{proof} 
We first note that any scaled version of $\mathbf{H}$, including $\mathbf{H}$ itself, can be expressed as follows 
\begin{equation*}{\tilde{\mathbf{H}}}=\begin{pmatrix} r_1c_1 h_{11}&r_1c_2  h_{12}& \mydots & r_1c_K  h_{1K}\\ r_2c_1  h_{21}&r_2c_2  h_{22}& \mydots & r_2 c_K  h_{2K}\\ \mydots & \vdots & \vdots & \vdots \\ r_K c_1  h_{K1}&r_K c_2  h_{K2}& \mydots & r_K c_K  h_{KK}\end{pmatrix}\hspace{-2.6pt}, \end{equation*}
where $r_i$ and $c_j$, for $i,j=1, \mydots, K$, are nonzero real numbers. The proof is effected by showing that $\tilde{\mathbf{H}}$ violates Condition~($\ast$) for all $r_i, c_j \in \mathbb{R} \text{\textbackslash{}}  \{ 0 \}$, $i,j=1, \mydots, K$. To this end, note that for all $r_i, c_j \in \mathbb{R} \text{\textbackslash{}}  \{ 0 \}$, $i, j=1, \mydots, K$, the following holds
\begin{equation*}
\tilde{h}_{ii}=h_{ii}r_{i}c_{i}= \frac{(h_{ij}r_{i}c_{j})(h_{ki}r_{k}c_{i})}{h_{kj}r_{k}c_{j}} \frac{a}{b} = \frac{ \tilde{h}_{ij} \tilde{h}_{ki} }{\tilde{h}_{kj}} \frac{a}{b},
\end{equation*} 
where $a, b \in \mathbb{Z}$ such that $\frac{a}{b}= \frac{h_{ii} h_{kj} }{h_{ij}h_{ki}}$. Since $a \tilde{h}_{ij} \tilde{h}_{ki} ,  b \tilde{h}_{kj} \in \mathcal{W}^{(\tilde{\mathbf{H}})}$ for all $r_i, c_j \in \mathbb{R} \text{\textbackslash{}}  \{ 0 \}$, $i, j =1 , \dots, K$, (\ref{rem3}) implies that Condition ($\ast$) is violated for all scaled versions of $\mathbf{H}$. Application of Theorem \ref{thm1} now yields the desired conclusion that $1/2$ DoF for each user cannot be achieved.  
\end{proof}
\end{theorem}

\section*{Appendices} \label{Appendix}

\subsection{Implications of requiring $1/2$ DoF for each user}
\label{subsection::ImplicationsOfRequiring1.2DoF}

We start with a definition needed in the formulation of the main result, Lemma \ref{lemma::1/2DoFMeansAtMostK/2Total} below.

\begin{definition} \label{def::totallyDisconnected} (Totally-disconnected users) 
We say that a user is totally disconnected if it does not experience interference from any other user and does not cause interference to any other user. Concretely, the $i$-th user, $i=1, \mydots, K$, is totally disconnected if  the $i$-th row and the $i$-th column of $\mathbf{H}$ have no nonzero off-diagonal elements. 
\end{definition}
\begin{lemma} \label{lemma::1/2DoFMeansAtMostK/2Total}
If each user in a $K \times K$ IC matrix $\mathbf{H}$, with $L$ totally disconnected users, is to achieve at least $1/2$ DoF, then the total number of DoF is exactly $L+ (K-L)/2$, where the totally-disconnected users achieve $1$ DoF each and the remaining users achieve exactly $1/2$ DoF each. \par
\begin{proof}
It follows directly that each totally-disconnected user achieves $1$ DoF as these users are interference-free. Now, consider a non-totally-disconnected user, say user $i$, $i \in \{1, \mydots, K\}$. Then, there exists  a distinct user $j  \in \{1, \mydots, K\} $
 which either experiences interference from user $i$ or causes interference to user $i$ or both. Next, consider the $2$-user IC $\mathbf{\tilde{H}}$ obtained by removing all users except for users $i$ and $j$ and note that $\text{DoF}_{i}(\mathbf{H}) \leqslant \text{DoF}_{i}(\mathbf{\tilde{H}})$ and $\text{DoF}_{j}(\mathbf{H}) \leqslant \text{DoF}_{j}(\mathbf{\tilde{H}})$ as
 the removed users simply constitute interference for users $i$ and $j$. Now, we know, thanks to \cite[Corollary 1]{Madsen05}, that in a $2$-user IC, the total number of DoF is bounded by 1, which together with
 $\text{DoF}_{i}(\mathbf{H}) \geqslant 1/2$ and $\text{DoF}_{j}(\mathbf{H}) \geqslant 1/2$, both by assumption, results in
 $\text{DoF}_{i}(\mathbf{H}) = \text{DoF}_{j}(\mathbf{H}) =1/2$.
 Summing over the $L$ totally-disconnected users and the $K-L$ non-totally-disconnected users yields $L+\frac{K-L}{2}$ DoF in total.
\end{proof}
\end{lemma}

\subsection{Preservation of individual DoF}
\begin{lemma} \label{lemma::EtkinPreserveDoF12}
For all IC matrices $\mathbf{H}$, the number of DoF achievable for each user is preserved under scaling according to Definition \ref{defScaling}.
\begin{proof}
The statement follows directly from \cite[Lemma 1]{Etkin09}. Specifically, \cite[Eqs. 2,3, and 4]{Etkin09} lead to the following conclusion: The capacity region of the IC with channel matrix $\mathbf{H}$ and that of any scaled version of $\mathbf{H}$ are asymptotically (in signal-to-noise ratio) identical. Therefore, the individual DoF, given by the pre-log factors of the corresponding individual rates $R_1, \mydots, R_K$, remain unchanged upon scaling of the underlying channel matrix.
\end{proof}
\end{lemma}

\subsection{Proof of Proposition \ref{prop}} \label{proofOfProp}

The proof is inspired by the proofs of \cite[Thm. 3]{Stotz16} and \cite[Thm. 4]{Wu15}. We first construct self-similar input distributions according to
\begin{equation*}  
\tilde{X}_j = \sum_{m= 0}^\infty V_{jm} r^m,
\end{equation*} 
where $r \in (0,1)$ and $\{V_{jm}: m \geqslant 0 \}$ is a sequence of i.i.d copies of $V_j$, $j=1, \mydots, K$. This ansatz is identical to that in \cite[Eq. (148)]{Wu15}, apart from the choice of the similarity ratio $r$. The following statements hold for the a.a. set of IC matrices $\mathbf{H}$, defined in \cite[Thm. 4]{Wu15} and denoted as $\mathcal{L}$ henceforth. Using \cite[Eq. (154)]{Wu15}, it follows that 
\begin{equation} \label{eqWuImport1}
d \Bigg ( \sum_{j} h_{ij} \tilde{X}_j \Bigg) = \frac{H \Big(\sum_{j} h_{ij} V_j  \Big )}{\log{(1/r)}},
 \end{equation}
 for $i=1, \mydots, K$. Similarly, thanks to \cite[Eq. (155)]{Wu15}, we obtain
\begin{equation} \label{eqWuImport2}
d \Bigg ( \sum_{j \neq i}  h_{ij} \tilde{X}_j \Bigg) = \frac{H \Big(\sum_{j \neq i} h_{ij} V_j  \Big )}{\log{(1/r)}}.
 \end{equation}
Combining (\ref{eqWuImport1}) and (\ref{eqWuImport2}), and using \cite[Eqs. 156, 157]{Wu15}, we get 
\begin{equation} \label{eqDavid60}
\sup_{V_1, \mydots, V_K} \left [ \frac{H \Big(\sum_{j} h_{ij} V_j  \Big )}{\log{(1/r)}}- \frac{H \Big(\sum_{j \neq i} h_{ij} V_j  \Big )}{\log{(1/r)}} \right ] \geqslant \text{DoF}_i
\end{equation}  
 where DoF$_i$ is as defined in (\ref{rateDoFDefinition}).  Finally, thanks to (\ref{eqWuImport1}), there exists an $r \in (0,1)$ such that   \begin{equation}  \label{eqDavid59}
 \log{(1/r)}= \max_{i=1, \mydots, K}\frac{H \Big(\sum_{j} h_{ij} V_j  \Big )}{d \Big ( \sum_{j} h_{ij} \tilde{X}_j \Big)} \geqslant \max_{i=1, \mydots, K}  H \Big(\sum_{j} h_{ij} V_j  \Big ),
 \end{equation}
where the inequality follows from the fact that the information dimension of a one-dimensional random variable cannot exceed 1 \cite[Eq. (13)]{Wu15}. Combining (\ref{eqDavid60}) and (\ref{eqDavid59}), and noting that (\ref{eqDavid59}) holds for $\sum_{j} h_{ij} V_j$ replaced by $\sum_{j \neq i} h_{ij} V_j$ as well, (\ref{prop1}) follows for all ICs in $\mathcal{L}$, and hence $\mathcal{L} \subseteq \mathcal{S}$. The proof is concluded upon noting that $\mathcal{L}$ is an a.a. set.

\subsection{Entropy growth} \label{AppendixTao}
\begin{lemma}\text{(A simple incarnation of \cite[Thm. 2.8.2]{Tao10})}:\label{lemtao}
Let $X, Y_1, \mydots, Y_m$ be discrete random variables, all of finite entropy, such that
\begin{equation} \label{tao1}
\frac{H(X+Y_i)}{H(X)} = 1+ \mathcal{O}(\varepsilon),
\end{equation}
for $\varepsilon \in (0, 1/2)$ and $i=1, \! \mydots ,m$. Then, for finite $m$, 
\begin{equation} \label{tao2}
 \frac{H(X+Y_1+\mydots + Y_m)}{H(X)} = 1+ \mathcal{O}(\varepsilon).
\end{equation}
\begin{proof}
In \cite[Th. 2.8.2]{Tao10}, set $\log{K_i}= H(X)\mathcal{O}(\varepsilon)$, for $i=1, \! \mydots, m$, and take the additive group $G$ to be $\mathbb{R}$. This results in $H(X+Y_1+\mydots + Y_m) \leqslant H(X)+ m H(X) \mathcal{O}(\varepsilon)$. Divide this inequality by $H(X)$ and note that owing to (\ref{wki2}) the expression on the LHS of (\ref{tao2}) is greater than or equal to 1. The proof is concluded by realizing that $\mathcal{O}(\varepsilon) m =\mathcal{O}(\varepsilon)$ thanks to $m$ being finite and independent of $\varepsilon$. 
\end{proof}
\end{lemma}

\subsection{Auxiliary lemma for the $K$-user case}

\begin{lemma} \label{lemmaAlmostAllSet} 
Let $\mathbf{H}$ be a $K \times K$ IC matrix. If $\mathbf{H}$ is in the a.a. set $\mathcal{S}$ covered by Proposition~\ref{prop} for DoF$_i=1/2$, $i=1, \mydots, K$, then all matrices obtained by scaling $\mathbf{H}$ are also in $\mathcal{S}$. 
\end{lemma}

\begin{proof}
We consider a scaled version of $\mathbf{H} \in \mathcal{S}$ according to 
\begin{equation*}{\tilde{\mathbf{H}}}=\begin{pmatrix} r_1c_1 h_{11}&r_1c_2  h_{12}& \mydots & r_1c_K  h_{1K}\\ r_2c_1  h_{21}&r_2c_2  h_{22}& \mydots & r_2 c_K  h_{2K}\\ \mydots & \vdots & \vdots & \vdots \\ r_K c_1  h_{K1}&r_K c_2  h_{K2}& \mydots & r_K c_K  h_{KK}\end{pmatrix}\hspace{-2.6pt}, \end{equation*} 
where $r_i, c_j$, $i,j=1, \mydots, K$ are nonzero real numbers. First, we show that $\mathbf{\tilde{H}} \in \mathcal{S}$. To this end, let, for an arbitrary $\varepsilon \in (0,1/2)$, $V_1, \mydots, V_K$ be random variables corresponding to $\mathbf{H}$ in (\ref{prop1}) for DoF$_i=1/2$, $i=1, \mydots, K$. Starting from
\begin{equation} \label{eqAlmostAllPreProof}
\frac{1}{2}- \varepsilon \leqslant  \frac { \Big [{ H \Big ({\sum _{j=1}^{K} h_{ij} V_j}\Big )- H \Big ({ \sum _{j\neq i}^{K}  h_{ij} V_j }\Big )}\Big ]}{ \max _{i=1, \mydots ,K} H   \Big ({\sum _{j=1}^{K}   h_{ij} V_j}\Big )},
\end{equation} 
we define the random variables $\tilde{V}_i := V_i/c_i, i=1,\mydots,K$. Applying Proposition \ref{prop:apply} with $\Psi_i=\tilde{V}_i$, for $i=1, \mydots, K$, and $r= 2^{-\max_{i=1, \dots, K} H(\sum_{j=1}^{K} h_{ij}\tilde{V}_j)}$, we conclude that user $i$ in $\mathbf{\tilde{H}}$ achieves
 \begin{equation} \label{eqAlmostAllProof}
  \frac { \Big [{ H \Big ({r_i  \sum _{j=1}^{K} h_{ij} V_j}\Big )- H \Big ({r_i  \sum _{j\neq i}^{K} h_{ij} V_j }\Big )}\Big ]}{ \max _{i=1, \mydots ,K} H   \Big ({r_i \sum _{j=1}^{K}   h_{ij} V_j}\Big )}  =  \frac { \Big [{ H \Big ({\sum _{j=1}^{K} h_{ij} V_j}\Big )- H \Big ({ \sum _{j\neq i}^{K}  h_{ij} V_j }\Big )}\Big ]}{ \max _{i=1, \mydots ,K} H   \Big ({\sum _{j=1}^{K}   h_{ij} V_j}\Big )} 
\end{equation}
DoF, where the equality in (\ref{eqAlmostAllProof}) holds because scaling a random variable does not change its entropy, and the choice of $r$ ensures that the nontrivial terms are selected in both minima on the RHS of (\ref{eq:apply}). Noting that (\ref{eqAlmostAllPreProof}) holds for all $\varepsilon \in (0,1/2)$ and combining it with (\ref{eqAlmostAllProof}) establishes that $\tilde{\mathbf{H}}\,\in\,\mathcal{S}$.
\end{proof}

\subsection{Entropy difference between linear combinations of random variables} \label{appEntropyDiff}

\begin{theorem} \cite[Theorem 14]{Wu15}.
Let $X$ and $Y$ be independent $G$-valued random variables, where $G$ denotes an arbitrary abelian group. Let $p, q \in \mathbb{Z} \setminus \{ 0 \}$. Then, 
\begin{align}H(p X + q Y) - H(X+Y) \leqslant& \hspace{1mm} \tau_{p,q}(2 H(X+Y) - H(X) - H(Y)), \end{align}
where $\tau_{p,q}= 7 \lfloor \log |p| \rfloor + 7 \lfloor \log |q| \rfloor +2$.
\end{theorem} 
\subsection{Entropy difference for i.i.d. random variables} \label{appEntropyIID}
We restate (a slight variation of) \cite[Th 3.5]{Madiman14}. For i.i.d. discrete random variables $X_1, X_2$,
\begin{equation} \label{EqEntrDiffApp}
\frac{1}{2} \leqslant \frac{I(X_1+X_2;X_2)}{I(X_1-X_2;X_2)} \leqslant 2.
\end{equation}
With $I(X;Y)= H(X)-H(X|Y)$, (\ref{EqEntrDiffApp}) becomes
\begin{equation*}
\frac{1}{2} \leqslant \frac{H(X_1+X_2) - H(X_1+X_2 | X_2) }{H(X_1-X_2) - H(X_1-X_2 | X_2) } \leqslant 2,
\end{equation*}
which is  equivalent to 
\begin{equation} \label{eqEntropyDifferenceIID}
\frac{1}{2} \leqslant \frac{H(X_1+X_2) - H(X_1) }{H(X_1-X_2) - H(X_1) } \leqslant 2.
\end{equation}

\subsection{A slight variation of \cite[Lem. 18]{Wu15}} \label{appWuMoreGeneral}
We formulate a slight variation of statement \cite[Eq. 284]{Wu15} in  \cite[Lem. 18]{Wu15}. 
\begin{lemma} \label{lastEverLemma}
 Let $X, X'$, and $Z$ be independent $G$-valued random variables, where $G$ is an abelian group and $X'$ has the same distribution as $X$. Let $p, r \in \mathbb{R}$. Then, 
\begin{equation*}
H(pX+Z) \leqslant H((p-r)X+rX'+Z)+ \Delta(X,X'),
\end{equation*}
where $\Delta(X,X')=H (X-X')-\frac{1}{2} H(X) -\frac{1}{2} H(X')$.
\begin{remark}
The only difference between  Lemma \ref{lastEverLemma} here and \cite[Lem. 18]{Wu15} is that \cite[Lem. 18]{Wu15} applies to $p, r \in \mathbb{Z} \setminus \{0\}$, whereas Lemma \ref{lastEverLemma} holds for $p,r \in \mathbb{R}$. Step-by-step inspection of the proof of  \cite[Lem. 18]{Wu15} reveals, however, that the result holds true more generally for $p, r \in \mathbb{R}$. 
 \end{remark}
\end{lemma}

\subsection{Restatement of \cite[Th. 8]{Wu15}} \label{4Th8restate}
\begin{theorem}\cite[Th. 8]{Wu15}.
Let $\mathbf{H}$ be a $3$-user IC matrix $\mathbf{H}$ with all off-diagonal entries algebraic numbers. If there exist distinct $i, j, k$ such that $h_{ij}, h_{ii}, h_{kj}$, and $h_{ki}$ are non-zero rational numbers, then the total number of DoF of $\mathbf{H}$ is strictly smaller than 3/2. 
\end{theorem}

\bibliographystyle{IEEEtran}
\bibliography{Zitat}

\end{document}